\documentclass[journal]{new-aiaa}

\usepackage[utf8]{inputenc}
\usepackage{textcomp}
\usepackage{graphicx}
\usepackage{amsmath}
\usepackage{isomath}
\usepackage{algorithm}
\usepackage{algpseudocode}
\usepackage{comment}
\usepackage[version=4]{mhchem}
\usepackage{siunitx}
\usepackage{longtable,tabularx}
\usepackage{multirow,multicol}
\usepackage{subcaption}
\usepackage{setspace}

\newtheorem{definition}{Definition}
\newtheorem{theorem}{Theorem}
\newtheorem{lemma}[theorem]{Lemma}

\newenvironment{proof}{\paragraph{Proof.}}{\hfill\ensuremath{\square}}

\title{An Evidential Reasoning Approach for Aerial Target Classification and Intent Prediction}

\author{Tenzing Thiley Bhutia \footnote{M.S Research Scholar, Department of Aerospace Engineering, tenzinthinley561@gmail.com.}, Subash Kumaraguru \footnote{Ph.D Candidate, Department of Aerospace Engineering,  subashkumara1999@gmail.com.}}
\author{Devaprakash Muniraj\footnote{Assistant Professor, Department of Aerospace Engineering, deva@smail.iitm.ac.in.}}
\affil{Department of Aerospace Engineering, Indian Institute of Technology, Madras-600036, India}
\begin{document}

\maketitle

\begin{abstract}
Timely classification and intent prediction of aerial targets is crucial for a combat aircraft to make informed tactical decisions. The prevailing approach for aerial target classification relies on data-driven models using time-series data. These models perform well with long-duration data; however, this is impractical in combat situations involving rapidly evolving threats that demand quick decisions. Minimizing false predictions is essential, as uncertainty is preferable to incorrect assessments in high-risk environments. Here, we propose an integrated approach to target classification and intent prediction that enables decisions from partial data in settings where threats require rapid response. In the proposed method, predictions are generated from short sequential sub-samples instead of the entire time series, and the results are refined by propagating beliefs across sub-samples. Outputs from classifiers are combined through an evidential reasoning framework to manage uncertainty. Target intent is  inferred using rule-based techniques and a distance-based combination method to fuse information over time.
Due to lack of publicly available datasets, a dataset for aerial target classification was generated for evaluation. A case study involving eight targets is used to demonstrate the effectiveness of the approach, whereby accuracies of 88\% and 93\% are achieved for target type classification and intent prediction, respectively. 
\end{abstract}

\section{Introduction}
\label{sec:intro}
Modern aerial combat environments are characterized by high-speed maneuvers, evolving threats, and the necessity for rapid decision making, thereby making accurate and timely identification of aerial targets, along with predicting their intent, a crucial requirement. For combat aircraft, maintaining situational awareness of surrounding aerial targets offers a decisive tactical edge during combat engagements. These targets can range from multi-role fighters and bombers to commercial aircraft and missiles, each potentially displaying a spectrum of behaviors, including hostile, non-hostile, or suspicious actions toward the ownship.

The problem of aerial target classification is addressed by analyzing prior knowledge and real-time information collected from various sensors on board the ownship. The input data is in the form of a time series and includes parameters of the target such as position, altitude, velocity, climb rate, turn rate, acceleration, etc. Together, these parameters create a vector-valued time series that is used to classify the identified target into one of several predefined categories of aerial targets. Each type of target - be it a highly maneuverable multi-role fighter, a large payload-carrying bomber, a civilian commercial aircraft, or a supersonic missile - exhibits distinct characteristics that aid in target classification. Some of these targets could exhibit hostile behavior towards the ownship. In this paper, we present an integrated approach for aerial target classification and intent prediction that employs evidential reasoning. The method involves a two-step hierarchical process. In the first step, the system classifies the target type, such as a multirole fighter or a civilian commercial airliner etc. The second step determines the target’s intent or activity by utilizing information about the identified target type and its relative position to ownship. 

The target classification module, which constitutes the first step, takes time-series measurements as input and generates class probabilities for the identified aerial targets. Formally, a time-series sample is represented by $T = \{p_1, p_2, p_3, \cdots, p_t\}$, where $t$ is the number of time steps and $p_i$ is the feature vector pertaining to time instance~$i$. The underlying application determines the nature and dimensionality of the feature vector.  
The task is modeled as a multi-dimensional time series classification problem that uses attributes such as target position \((x, y, z)\), velocity \((u_x, u_y, u_z)\), acceleration \((a_x, a_y, a_z)\), turn rate \((\dot{\phi})\), and radar cross-section (RCS) to classify the target type. The predictions are done on sub-samples of small duration and an ensemble classification approach using multiple classifiers is employed. A belief distribution is computed from the output of each classifier, and the beliefs are combined across classifiers as well as sub-samples using a modified belief combination technique. The target intent prediction module, representing the second step, identifies the intent of the target relative to the ownship. Belief predictions from the target classification module for each sub-sample are first converted into membership values, enabling categorization of the target based on its capabilities and typical mission profiles. Using a set of predefined rules, new labels are assigned by integrating target type information with target behavior. A standard rule-based decision tree classifier is then used to classify target intent. Finally, a distance-based combination technique is proposed to combine the belief distributions over time, providing a robust prediction of target type and intent.

We demonstrate the effectiveness of this approach through a case study that involves classifying eight different aerial targets exhibiting three types of behavior (intent) towards the ownship: hostile, non-hostile, and suspicious. The experimental results show a classification accuracy of 88\% for target type identification and 93\% for target intent prediction. Image-based target type classification \cite{wu2020benchmark,zhang2022information} in combat is limited by sensor constraints, occlusions, and the difficulty of generating realistic datasets. Time-series kinematic data captures dynamic target behavior more reliably. Hence, to the best of our knowledge, this study is the first to implement a time-series-based approach using target kinematic and dynamic features for target type classification in combat environments.
The key contributions of this work are as follows:
\begin{enumerate}
     \item We introduce a time-series-based approach using kinematic features for target type classification, which, to the best of our knowledge, has not been explored in prior works. 
     \item We proposed an early classification framework that utilizes sub-samples of the target’s kinematic trajectory data, rather than relying on the entire time series as is common in the existing works on target classification. This approach enables faster decision-making and helps reduce the pilot’s reaction time. 
     \item We developed an ensemble approach for aerial target classification that integrates predictions from multiple standard machine learning (ML) classifiers as well as across time. 
     We present a new approach to reduce misclassification when combining predictions from multiple classifiers and across time. When the predictions are conflicting, rather than making a potentially incorrect classification, our method redistributes belief to a set of classes that reflect uncertainty in the prediction. This ensures that when the evidence is inconclusive, the approach makes a prediction with uncertainty instead of making an erroneous classification.
     \item We proposed a rule-based classifier to predict the intent of aerial targets, categorizing them as hostile, non-hostile, or suspicious. This module leverages target type predictions from the target classification module, along with the kinematics of the target relative to the ownship.
     \item In combat scenarios where target intent can shift rapidly, the well-established Dempster’s rule of combination often struggles to handle highly conflicting evidence effectively. To address this limitation, we proposed a combination technique that incorporates a distance metric to quantitatively assess the degree of conflict between sources. When this metric exceeds a predefined threshold, our method adaptively modifies the combination rule to better manage the conflict. By doing so, we prevent the propagation of incorrect predictions over time, thereby resulting in reliable information fusion in dynamic environments. 
    
\end{enumerate}

The remainder of this paper is organized as follows. In Section \ref{sec:related}, we review related works on time-series classification and combination techniques for aerial target classification and other related applications. Section \ref{sec::classification} provides background information on time-series classification and combination techniques. Section \ref{sec:proposed_approach} describes the proposed integrated target classification and intent prediction approach. 
Section \ref{sec:data_generation} outlines the dataset used in this paper, along with the procedures for data generation and preprocessing. Section \ref{sec::experimental_evaluation} presents the findings from the experimental evaluation. Finally, Section~\ref{sec:conclusion_and_future_works} discusses conclusions and suggestions for future work.

\section{Related Work}
\label{sec:related}
Existing research on time series classification for aerial applications can be broadly categorized based on the primary classification objective. Authors in \cite{zhang2022information,wu2020benchmark} focus on classifying aerial target types using image features, while \cite{wang2023quick,zhang2024target,wang2024intelligent,singh2014dynamic} emphasize classifying the intent of aerial targets using kinematic features. These studies, however, have been conducted under diverse experimental conditions and utilize non-publicly accessible time-series datasets, thereby limiting reproducibility. \cite{wu2020benchmark,zhang2023combat} leverage images as features for target type classification. However, image-based methods can be unreliable in combat environments due to limitations of sensors, issues with occlusions, and the challenges of generating realistic data under dynamic conditions.
 Despite the progress made in these distinct areas, there is a notable absence of an integrated framework that jointly addresses target classification and intent prediction. This gap highlights the need for a comprehensive approach that can effectively unify these two critical aspects of aerial target classification.
 
In aerial target classification, there could be inherent uncertainty due to missing or ambiguous information from multiple sensors \cite{frey2018f}. Further, the time taken for classification given an input time-series, called the {\textit{response time}, should be low for time-critical applications such as air combat. In the present work, we address the aforementioned issues and present an integrated aerial target classification and intent prediction algorithm. Classification, when performed on partial time-series data, is known as \textit{early classification} \cite{gupta2020approaches}. A popular approach to solving the early classification problem is through the use of shapelets \cite{ye2009time}. Shapelet transform \cite{hills2014classification,aldhanhani2019framework} converts time-series data into local-shape pattern space, with each feature denoting the pattern match score/distance between a time series and a shapelet. Other approaches for early classification include dictionary-based approaches \cite{schafer2015boss} that transform raw data to another domain based on frequency, using Fourier transform, Cosine transform, etc. However, these frequency-domain and shapelet-based methods are computationally expensive and are not suitable for onboard implementation on an aircraft. By making use of sub-samples rather than a time series of long duration for prediction, the proposed approach makes an early prediction with only a short time duration of time-series data, and the predictions are continuously refined as new sub-samples are fed to the algorithm. In other words, the prediction at the current time not only considers the most recent sub-sample but also information from all the past sub-samples. 

Approaches that employ standard ML methods have been extensively applied for target classification. Authors in \cite{SOUZA2018198} make use of instance-based learning, such as K-Nearest Neighbour (KNN), where the class label is classified based on a similarity or distance measure. Though Dynamic Time Warping (DTW) and its variants have been shown to be the best method theoretically, it is computationally demanding and might not be suitable for near-real-time applications such as aerial target classification. Recently, decision trees \cite{visser2003using} and their variants have been extensively used in predicting the behavior or strategies of one or more opponents in adversarial environments. Ginoulhac et al. \cite{ginoulhac2019target} used the gradient boosting algorithm to classify aerial targets. While the approach presented in this work is agnostic to the specific type of classifier, herein, Feed Forward Neural Networks (FFNN) \cite{fukushima2018online}, Support Vector Machine (SVM) \cite{dihel2023classifying,hui2017dempster}, Long Short Term Memory (LSTM) \cite{zhang2022information,cemiloglu2025handling} and Decision Tree \cite{gupta2019rule} have been used for their simplicity and complementary advantages. The output from each classifier is a probability distribution over the set of target classes. 

To handle uncertainty in the classifier prediction, each classifier is designed to generate a probability distribution over a subset of the power set of the output classes instead of a single class point estimate. These distributions are then combined using a combination rule based on Dempster-Shafer (DS) theory \cite{zhang2024target,kessentini2015dempster,oukhellou2010fault} to create a final composite belief. Since Dempster's combination rule is found to produce counterintuitive results due to conflict among the sources, Yager's rule \cite{yager1987dempster} has been proposed to redistribute the conflicting beliefs to the overall uncertainty, offering a more robust combination strategy. However, this often leads to a lack of actionable insights when it comes to applications such as aerial target classification. In the approach proposed herein, the conflict is reallocated to specific focal sets representing distinct hypotheses, as multiple aerial targets may share similar features. This approach aids in better decision-making by shifting the uncertainty to the focal sets while highlighting uncertainties. In situations where a target's intent evolves over time, naive application of existing belief combination rules will lead to incorrect conclusions due to the combination of conflicting temporal evidence. To address this issue, we propose a method that utilizes a distance metric \cite{jousselme2001new} to quantify the differences between the evidence. If the distance metric remains within a predefined range, belief combination continues, improving belief updates and enhancing confidence in target intent prediction. If the metric exceeds a predefined threshold, belief combination is stopped and a new evidence set is initialized to prevent misclassification. By making use of early prediction from partial data, and evidence combination across multiple sources and across time, the proposed approach strikes a balance between response time and accuracy of prediction. Military target intention recognition studies \cite{wang2023quick,zhang2024target,singh2014dynamic,wang2024intelligent} use highly specialized, custom datasets created from different simulated battle environments with unique features and targets. Each dataset differs greatly in how data is collected, the features measured, and how target behaviors or intents are defined and labeled. Due to these differences, models designed for one dataset heavily depend on its specific input features, data formats, and label definitions. At the same time, there is no common, publicly available benchmark dataset used across these studies, making direct comparisons difficult. Thus, prior work in this area \cite{wang2023quick,zhang2024target,singh2014dynamic,wang2024intelligent} typically involves comparing the proposed approach only against baseline methods adapted or re-implemented for their own specific datasets, rather than directly comparing it against other works. 

\section{Preliminaries}
\label{sec::classification}
\subsection{Dempster-Shafer Theory} \label{prelim_DStheory}
Dempster-Shafer Theory (DST) is a mathematical framework for analyzing evidence. The foundational work on this subject is detailed in \cite{shafer1976mathematical}. In traditional probability theory, evidence is linked to a single possible event. In contrast, DST allows evidence to be associated with multiple potential events, such as sets of events. This means that evidence in DST can operate at a higher level of abstraction without requiring specific assumptions such as the events to reside within the evidential set. When the evidence is sufficient to assign probabilities to individual events, the Dempster-Shafer model simplifies and aligns with traditional probability theory. A key feature of DST is its ability to handle varying levels of precision in information, without necessitating additional assumptions for representation. 
There are three important functions in DST: the basic belief assignment function (\(bba\) or \(m\)), the Belief function \(Bel\), and the Plausibility function \(Pl\). Since the proposed approach relies on these functions, they are defined in the sequel for completeness. \\
Let $\Omega$ be the universe of all the $N$ hypothesis, i.e., $\Omega=\left\{X_1, X_2 ,\ldots, X_N\right\}$. $\Omega$ is also called the frame of discernment. 
Given $\Omega$, we can define a power set $2^{\Omega}$ that contains all possible subsets of $\Omega$, including the empty set $\emptyset$.

\begin{definition}[Basic Belief Assignment $(BBA)$]
The theory of evidence assigns to each element of the power set $2^{\Omega}$, a belief mass, denoted by $m$, called basic belief assignment. Formally, $m$ is defined as
\begin{equation}
m: 2^{\Omega} \rightarrow[0,1], \ \textrm{such that} \quad m(\emptyset)=0 \quad \textrm{and} \quad
\sum_{A \in 2^{\Omega}} m(A)=1. \label{eqn:bba_prop}
\end{equation} 
\end{definition}

Given two BBAs $\boldsymbol{m}_1$ and $\boldsymbol{m}_2$, the combined belief, denoted as $m_1 \underset{DS}{\oplus} m_2$, can be written using Dempster's rule of combination as
\begin{equation}
m_1 \underset{DS}{\oplus} m_2 = \frac{\sum_{B \cap C=A \neq \emptyset} m_1(B) \cdot m_2(C)}{1-K}, \ \forall A \subseteq \Omega, \  \textrm{where} \ K=\sum_{B \cap C=\emptyset} m_1(B) \cdot m_2(C).
\label{eqn:dempster_combi}
\end{equation} 
In the preceding expression, $K$ is a measure of conflict between the two given mass sets.  However, Dempster's rule will produce counter-intuitive results when the conflict is high. To overcome this issue, the Yager combination
rule \cite{yager1987dempster} was proposed, whereby the conflicting belief is moved to $\Omega$ as shown below
\begin{equation}
\begin{aligned}
m_1 \underset{YG}{\oplus} m_2(A) & =\sum_{B \cap C =A} m_1\left(B\right) m_2\left(C\right) \ \textrm{for}\ A \neq \emptyset \\
m_1 \underset{YG}{\oplus} m_2(\Omega) & =\sum_{B \cap C=\Omega} m_1\left(B\right) m_2\left(C\right)+K, \ \textrm{where}\  K =\, \sum_{B \cap C=\emptyset} m_1\left(B\right) m_2\left(C\right).
\end{aligned}
\label{eqn:yager_combi}
\end{equation}

A common method for quantifying the degree of conflict between pieces of evidence is through the evidence distance, which is defined as follows.

\begin{definition}[Evidence Distance]
Suppose that $m_1$ and $ m_2$ are two BBAs defined on $\boldsymbol{\Theta}$, then the evidence distance between $m_1$ and $m_2$ is defined as 
\begin{equation}
d\left(m_1, m_2\right)=\sqrt{\frac{1}{2}\left(m_1-m_2\right)^T D\left(m_1-m_2\right)}
\label{eqn:evidence_distance}
\end{equation}
where $\boldsymbol{m}_1$ and $\boldsymbol{m}_2$ are column vectors of length $2^{\boldsymbol{N}}$ corresponding to $\boldsymbol{m}_1$ and $\boldsymbol{m}_2$, respectively, and $N$ is the size of the set $\boldsymbol{\Theta}$. $D$ is a $2^N \times 2^N$ matrix whose elements are defined as
$$
D(A,B)=\frac{|A \cap B|}{|A \cup B|}, \ \textrm{where} \ A \subseteq \Theta  \ \text{and} \  B \subseteq \Theta.
$$
\end{definition}

The transferable belief model (TBM) represents beliefs using belief functions or BBAs instead of classical probability distributions \cite{smets2008transferable}. However, when a decision must be made, it is desirable that these beliefs are expressed as probabilities so that they can be integrated into the decision-making frameworks. This conversion is achieved by the pignistic transformation, which converts the BBA into a probability distribution known as the pignistic probability. 

\begin{definition}[Pignistic Probability]
Give a basic belief assignment $\boldsymbol{m}$, its corresponding pignistic probability, denoted by $betp$, is defined as 
\begin{equation}
betp(\{\omega\})=\sum_{A \subseteq \Omega, \ \omega \in A} \frac{m(A)}{|A|(1-m(\emptyset))}, \  \forall \omega \in \Omega 
\label{eqn:betp}
\end{equation}
where \(A\) is a \textit{focal set}, \(|A|\) is its \textit{cardinality}, and \((1 - m(\emptyset))\) is a \textit{normalization factor} that accounts for conflict. \(\omega\) is a \textit{singleton element} in \(\Omega\).
\end{definition}

\subsection{Fuzzy set theory}
Fuzzy set theory, introduced in \cite{zadeh1965fuzzy}, extends classical set theory by incorporating the principle of partial membership, allowing elements to belong to a set in varying degrees. This enables effective handling of uncertainty and imprecision, which is especially valuable in engineering applications where precise reasoning and sharp classifications are often impractical.

A fuzzy set $\tilde{A}$ over a universe of discourse $X = \{x_1, x_2, \ldots, x_n\}$ is defined through a membership function $\mu_{\tilde{A}}$ and is represented as
\begin{equation}
\tilde{A} = \{ \langle x, \mu_{\tilde{A}}(x) \rangle \mid x \in X \}, \ \textrm{where} \ \mu_{\tilde{A}}: X \to [0, 1]. 
\label{eqn:fuzzy_logic}
\end{equation}
$\mu_{\tilde{A}}(x)$ indicates the degree to which each element $x \in X$ belongs to the fuzzy set $\tilde{A}$.

\section{Proposed Approach}
\label{sec:proposed_approach}
In aerial target classification and intent prediction, a key requirement is to classify the target and infer its intent with a low response time, providing timely and actionable information to the other subsystems. Needless to say, minimising false predictions is essential to reduce operational risk and enhance decision-making accuracy. This section outlines the proposed methodology for aerial target classification and intent prediction.
The methodology described here addresses aerial target classification and intent prediction as a time-series classification problem. Conventional time-series classification methods outlined in Section \ref{sec:related}, typically require analyzing longer temporal windows to make accurate predictions. However, this increases response times, delaying critical decisions, which is something undesirable in aerial target classification and intent prediction.

The proposed approach leverages the idea of early classification. Instead of analyzing the complete time series, predictions are made on shorter subsamples. These early predictions by themselves could be less accurate due to lack of sufficient data, therefore, the predicted beliefs from the subsamples are propagated over time using an evidential reasoning framework. Thus, the proposed approach enables faster predictions supporting rapid responses, improved accuracy and confidence in the prediction over time as the belief is continually updated with subsequent evidences. The proposed methodology balances the trade-off between rapid response and prediction accuracy, making it well-suited for dynamic environments demanding near-real-time aerial target classification and intent prediction. The proposed approach consists of two main modules: 
\begin{itemize}
    \item[1] Target type classification.
    \item[2] Target intent prediction.
\end{itemize}
 The schematic of the proposed approach, including the two modules, is shown in Fig.~\ref{fig:flow-diagram}.
\subsection{Target Type Classification} 
The target type classification module classifies the aerial targets into one of the eight predefined class of aerial targets, including multi-role fighter, small unmanned aircraft, bomber, etc. using kinematic input features such as altitude, velocity, acceleration, turn rate, etc.
\begin{figure}[htbp]
    \centering
    \includegraphics[width=0.8\textwidth]{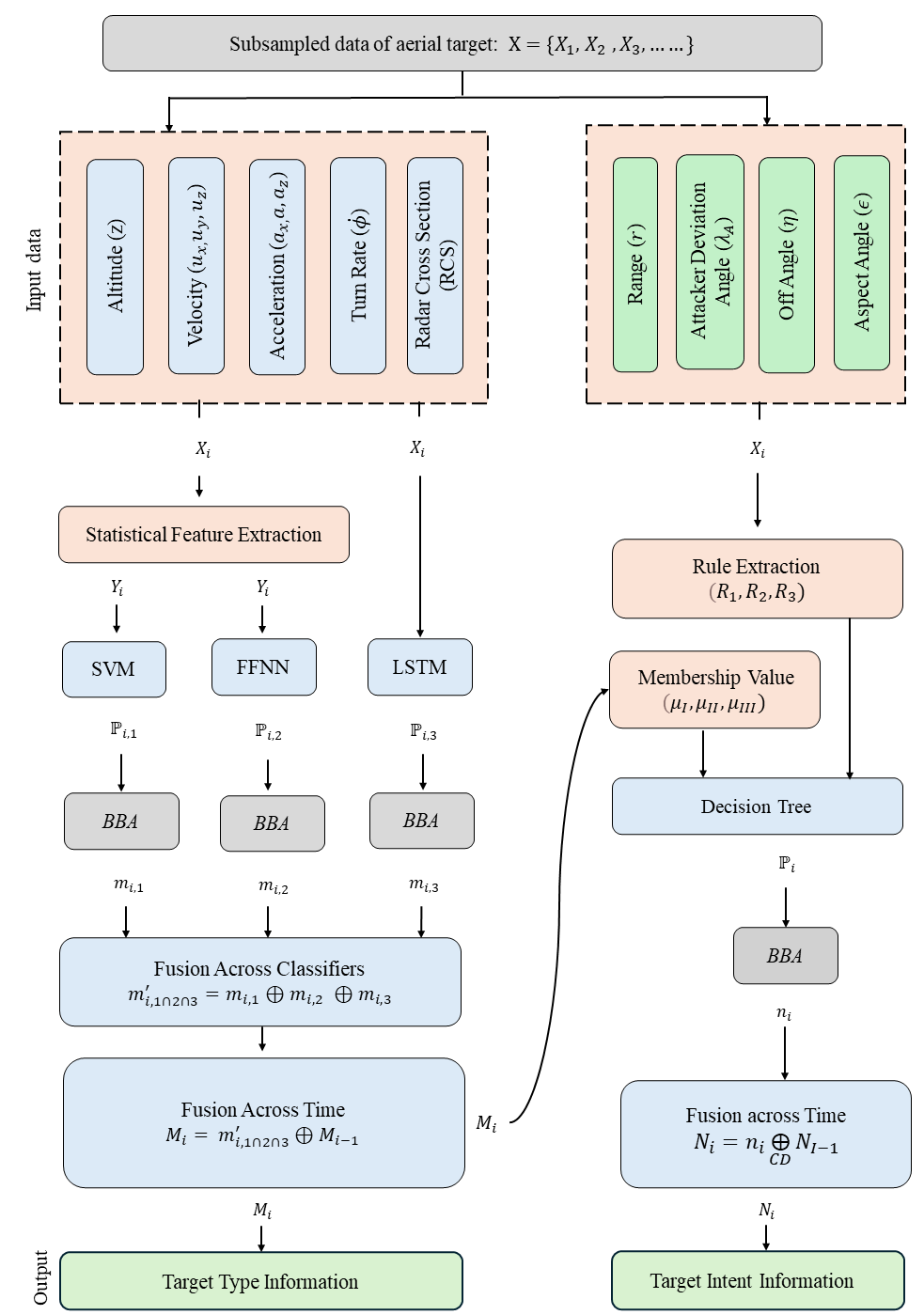}
    \caption{The schematic of the proposed integrated target classification approach.}
    \label{fig:flow-diagram}
\end{figure} 
The input to the module is a time-series subsample of window size $w$, i.e., the subsample has $w$ real-valued observations. Each observation consists of $k$ distinct measurements (features) obtained from sensors such as accelerometers, gyroscopes, radar, etc. Formally, a time-series $X = \{X_1, X_2, X_3,\ldots\}$ is represented as an ordered collection of subsamples. Each subsample, $X_i$, consists of $w$ ordered real-valued observations, i.e., $X_{i} = [x_{i, 1}, x_{i, 2}, \ldots, x_{i, w}]$. We assume that the measurements from various sensors are fused, and the time-synchronized measurements are available to the target type classification module as a discrete-time signal. For instance, $x_{i,j}$ denotes the measurements (features) corresponding to the time instant $t_{i,j}$, where $t_{i,j} = (i-1)w\Delta t + (j-1) \Delta t$ for $i=1,2,\ldots$ and $j=1,2,\ldots, w$, and $\Delta t$ is the sampling time. At each time step, there are $k$ distinct features, i.e., $x_{i,j} \in \mathbb{R}^k$, where $i = 1,2,\ldots$ and $j = 1,2,\ldots,w$. Thus, a subsample $X_i$ is a matrix of size $k \times w$. It is noted that the target type classification module utilizes only a subset of the $k$ features, specifically, the first $h$ features, where $h < k$. The rest of the $(k-h)$ features are employed in the target intent prediction module, which will be described in the next section. For each instance or subsample $X_i$, the target type classification module predicts a target class label $c_i$, where $c_i \in \Omega$ and $\Omega$ is the set of target type class labels.

The target type classification module, depicted in the left half of Fig.~\ref{fig:flow-diagram}, takes the subsample $X_i$ as input. The subsample includes time-synchronized features such as altitude $(z)$, velocity $(v_x,v_y,v_z)$, acceleration $(a_x,a_y,a_z)$, turn rate $(\dot\phi)$ and RCS of the target. The Statistical Feature Extraction block transforms each subsample into a compact representation using statistical descriptors (e.g., mean, standard deviation, skewness). This step reduces the dimensionality of the subsample while preserving key temporal patterns, thereby reducing the computational effort involved in learning and inference. The subsample or the extracted features are fed into an ensemble classifier consisting of a bank of classifiers. While the proposed approach is general and can be applied to different types of classifiers, in this work, we utilize three classifiers: Support Vector Machine (SVM), Feedforward Neural Network (FFNN), and Long Short-Term Memory (LSTM). The SVM and FFNN work on the reduced statistical features, capturing discriminative patterns, while the LSTM directly processes the subsample to capture temporal dependencies within the subsample. Each classifier provides a probability distribution over the set of target classes. These probability outputs are subsequently converted into BBA using a mapping that explicitly incorporates uncertainty in prediction. The Fusion Across Classifier block combines BBAs from multiple classifiers to provide a fused BBA for the current subsample. The resulting BBA is then fused with the BBA corresponding to the previous subsample in the Fusion Across Time block to result in the final BBA. It is emphasized that the final BBA resulting from the Fusion Across Time block not only uses information from the current subsample for prediction, but all the previous subsamples through the combination of the BBAs across time. In the subsequent paragraphs, each of the blocks depicted in Fig.~\ref{fig:flow-diagram} is explained in detail, starting with the Statistical Feature Extraction block. 

\subsubsection{Statistical Feature Extraction} \label{ss:class_features}

For computational efficiency, in some classifiers such as the FFNN and SVM, instead of directly using the raw features or subsample, we make use of statistical features such as mean, variance, skewness, kurtosis, etc. that are derived from each subsample. This reduction in dimensionality results in a matrix representation of size $h \times p$ for each subsample, where $p$ represents the number of statistical features and $p<w$. To extract compact and discriminative features, we choose 11 statistical descriptors independently for each feature dimension over the $w$ time steps. Given the subsample $X_i$, \(x_{i,j}^{(f)} \) denotes the feature \( f \in \{1, 2, \ldots, h\} \) at measurement $j$ in subsample $i$. Given the measurements \(x_{i,j}^{(f)} \) for $j=1,\ldots,w$, let \(\bar{x}_{i,j}^{(f)} \) for $j=1,\ldots,w$ denote the ordered set of measurements that are ordered from lowest to highest. The chosen statistical features along with their mathematical description are presented in Table \ref{tab:statistical_features}. In the sequel, $\lfloor m \rfloor$ denotes the integer part of $m$. The output of the Statistical Feature Extraction block is denoted by $Y_i$, where $Y_i$ is a matrix of size $h \times p$. 

\renewcommand{\arraystretch}{1.4}
\begin{table}[h]
    \centering
    \caption{Statistical Features Used in Target Type Classification}
    \begin{tabular}{ll}
        \hline \hline
        Feature Name & Mathematical Description \\
        \hline
        Maximum                       & $\bar{x}_{i,w}^{(f)}$ \\
        Minimum                       & $\bar{x}_{i,1}^{(f)}$ \\
        \multirow{2}{*}{Median}                  & $\bar{x}_{i,\frac{(w+1)}{2}}^{(f)}$, if $w$ is odd \\ & $\frac{1}{2}\big(\bar{x}_{i,\frac{w}{2}}^{(f)} + \bar{x}_{i,\frac{w}{2}+1}^{(f)}\big)$, if $w$ is even \\
        Mean ($\mu_i^{(f)}$) & $\frac{1}{w} \sum_{j=1}^{w} \bar{x}_{i,j}^{(f)}$\\
        Standard deviation ($\sigma_i^{(f)}$) & $\sqrt{ \frac{1}{w} \sum_{j=1}^{w} \left( x_{i,j}^{(f)} - \mu_i^{(f)} \right)^2 }$\\
        Range & $\bar{x}_{i,w}^{(f)} - \bar{x}_{i,1}^{(f)}$\\
        Skewness & $\frac{1}{w} \sum_{j=1}^{w} \left( \frac{x_{i,j}^{(f)} - \mu_i^{(f)}}{\sigma_i^{(f)}} \right)^3$\\
        Kurtosis & $\frac{1}{w} \sum_{j=1}^{w} \left( \frac{x_{i,j}^{(f)} - \mu_i^{(f)}}{\sigma_i^{(f)}} \right)^4$\\
        First quartile ($Q_1$) & $(1-\alpha)\,\bar{x}_{i,m}^{(f)} + \alpha \,\bar{x}_{i,m+1}^{(f)}$, where $m = \lfloor \frac{w+1}{4} \rfloor$ and $\alpha = \frac{w+1}{4} - \lfloor \frac{w+1}{4} \rfloor$\\
        Third quartile ($Q_3$) & $(1-\alpha)\,\bar{x}_{i,m}^{(f)} + \alpha \,\bar{x}_{i,m+1}^{(f)}$, where $m = \lfloor \frac{3(w+1)}{4} \rfloor$ and $\alpha = \frac{3(w+1)}{4} - \lfloor \frac{3(w+1)}{4} \rfloor$\\
        Inter-quartile range (IQR) & $Q_3 - Q_1$ \\          
        \hline \hline
    \end{tabular}
    \label{tab:statistical_features}
\end{table}

\subsubsection{Ensemble Classification and Belief Mapping}
\label{ss:class_mapping}
Once the statistical features $Y_i$ are extracted, $X_i$ and $Y_i$ are fed to the bank of classifiers. Classifiers such as SVM and FFNN take the extracted statistical features $Y_i$, while LSTM models directly process the subsample $X_i$ to capture temporal dependencies. The classification process is a mapping that takes either $X_i$ or $Y_i$ as input and gives the probability distribution over the set of target type class labels as output. Suppose there are $q$ target types \(c_1,c_2,\ldots,c_q\) to be classified, the frame of discernment $\Omega$ can be constructed as \(\Omega = \{c_1,c_2,\ldots,c_q\)\}. The probability distribution from classifier $\mathbb{C}_j$ for subsample $X_i$ is given as \(\mathbb{P}_{i,j}(\Omega) = \{P_1^{i,j},P_2^{i,j},\ldots,P_q^{i,j}\}\). These probability distributions $\mathbb{P}_{i,j}$ must be effectively combined to establish a comprehensive prediction. To combine the evidences ($\mathbb{P}_{i,j}$) from the classifiers, we employ Dempster-Shafer theory of evidential reasoning. Unlike the Bayesian approach \cite{verbert2017bayesian} that requires prior probabilities and can only assign belief to singleton hypotheses, DS theory allows for modeling incomplete, uncertain, or ambiguous information through Basic Belief Assignment as discussed in Section~\ref{prelim_DStheory}. This enables combining beliefs over both singleton and composite class subsets.

To leverage the benefits of evidential reasoning provided by DS theory, we first need to convert the probability distribution $\mathbb{P}_{i,j}$ from each classifier into a belief distribution $\overline{m}_{i,j}$. To that end, we introduce an additional variable $\mathit{err\_rate}_j$, which represents the complement of the accuracy of classifier $\mathbb{C}_j$, i.e., 
\begin{equation}
\mathit{err\_rate}_j = 1 - {accuracy}_{\mathbb{C}_j}.
\label{eqn:error_rate}
\end{equation}
The error rate indicates uncertainty or in other words, belief over the complete frame of discernment. The map from $\mathbb{P}_{i,j}$ to $\overline{m}_{i,j}$ accounting for the error rate is defined as given below:
\begin{equation}
\overline{m}_{i,j}(c)=
\begin{cases}
    0, \ \text{if} \ c = \emptyset \\
    (1- \mathit{err\_rate}_j) \times P_1^{i,j},\ \text{if} \ c = c_1\\
    (1- \mathit{err\_rate}_j) \times P_2^{i,j},\ \text{if} \ c = c_2 \\
    \vdots \\
    (1- \mathit{err\_rate}_j) \times P_q^{i,j},\ \text{if} \ c = c_q \\
    \mathit{err\_rate}_j,\ \text{if} \ c = \Omega
\end{cases}.
\label{eqn:bba}
\end{equation}
\noindent Let $\overline{\mathcal{F}} = \{c_1, c_2, \ldots, c_q, \Omega, \emptyset\}$ denote the domain of $\overline{m}_{i,j}$, where $\overline{m}_{i,j}(\Omega)$ represents the uncertainty in prediction of classifier $\mathbb{C}_j$. Certain groups of target classes often exhibit highly similar characteristics, resulting in considerable classification ambiguity due to overlapping features, as illustrated in Table~\ref{tab:targetTypes}. For instance, \texttt{Bomber} and \texttt{Civilian commercial airliner} possess comparable kinematic and dynamic characteristics, with even their RCS values falling within similar ranges. Similarly, \texttt{Military helicopter} and \texttt{Civilian helicopter} can be challenging to distinguish, as they share closely related features that make it difficult for the classifiers to discriminate between them. In combat scenarios, providing an uncertain prediction is often more beneficial to the pilot than an incorrect one. For instance, if the actual target is a \texttt{Bomber} but the classifier erroneously identifies it as a \texttt{Civilian commercial airliner}—owing to the challenge of distinguishing between these similar classes—this misclassification could have serious consequences. In safety-critical applications such as the one addressed in this study, it is preferable for the target type prediction algorithm to communicate uncertainty, indicating that the target may be either a \texttt{Bomber} or a \texttt{Civilian commercial airliner}, until additional evidence is available to make a confident and accurate identification. This approach prioritizes operational safety by preventing potentially hazardous decisions based on unreliable predictions. 

Let \(\theta_x\ \subseteq 2^\Omega \ \text{for} \ x = 1,2,\ldots,l\) denote a predefined set of focal elements representing target types that share similar characteristics and thereby could potentially lead to erroneous predictions as discussed in the previous paragraph. In this work, we assume that a target type can be part of at most one focal set. The definition of $\theta_x$ would be problem-specific and can either be inferred apriori by an expert or during the training process. For example, \(\theta_l = \{c_2,c_3\}\) represents a two-element set consisting of targets $c_2$ and $c_3$ with similar characteristics and RCS signatures. An appealing feature of the evidential reasoning framework is that such focal sets can be explicitly included in the belief prediction as outlined below. We modify the BBA $\overline{m}_{i,j}$ presented in (\ref{eqn:bba}) to account for the focal sets $\theta_x$ for $x = 1,2,\ldots,l$ to result in the belief distribution $m_{i,j}$ as shown below. With the $l$ focal sets $\theta_x$, the belief distribution $m_{i,j}$ will have $(q+l+2)$ elements, which includes $q$ singleton sets, $l$ focal sets, the empty set $\emptyset$, and the frame of discernment $\Omega$. The belief distribution $m_{i,j}$ is defined as 
\begin{equation}
m_{i,j}(c) = 
\begin{cases}
0, & \text{for } c = \emptyset \\
\overline{m}_{i,j}(c), & \text{for } c \notin \theta_x \ \text{and} \ c \neq \Omega  \text{ for } x = 1, 2, \ldots, l \\
\alpha_{i,j}^x \times \overline{m}_{i,j}(c), & \text{for } c \in \theta_x \text{ for } x = 1, 2, \ldots, l \\
(1 - \alpha_{i,j}^x) \times \sum_{c_k \in \theta_x} \overline{m}_{i,j}(c_k), & \text{for } c = \theta_x \text{ for } x = 1, 2, \ldots, l \\
\overline{m}_{i,j}(c), & \text{for } c = \Omega 
\end{cases}.
\label{eqn:bba_updated}
\end{equation}
In the preceding expression, $\alpha_{i,j}^x \in [0,1] $ is a weighting term given by 
\begin{equation}
\alpha_{i,j}^x=\frac{\max _{c_p \in \theta_x, c_q \in \theta_x, c_p \neq c_q}\left|\overline{m}_{i,j}\left(c_p\right)-\overline{m}_{i,j}\left(c_q\right)\right|}{\sum_{c_k \in \theta_x} \overline{m}_{i,j}\left(c_k\right)}\,.
\end{equation}
Let $\mathcal{F} = \{c_1, c_2, \ldots, c_q, \theta_1, \theta_2, \ldots, \theta_l, \Omega, \emptyset\}$ denote the domain of $m_{i,j}$. In the following lemma, we show that the updated belief distribution \( m_{i,j} \) defined in \eqref{eqn:bba_updated} results in a valid BBA that satisfies the non-negativity and normalization properties as outlined in \eqref{eqn:bba_prop}. 

\begin{lemma}
The belief distribution \( m_{i,j} \) defined in \eqref{eqn:bba_updated} results in a BBA that satisfies the following properties: 
\begin{equation*}
m_{i,j}(A) \in [0,1] \ \forall A \in \mathcal{F}, \ m_{i,j}(\emptyset)=0, \ \textrm{and} \ 
\sum_{A \in \mathcal{F}} m_{i,j}(A)=1.
\end{equation*} 
\end{lemma}

\begin{proof}
We begin by showing that $m_{i,j}(\emptyset)=0$ and $m_{i,j}(A) \in [0,1] \ \text{for all} \ A \in \mathcal{F}$. We observe that $m_{i,j}(A) \in [0,1]$ for $A = \emptyset$ and $ A = \Omega$ since $m_{i,j}(\emptyset) = 0$ and $m_{i,j}(\Omega) = \overline{m}_{i,j}(\Omega) = err\_rate_j$, where $err\_rate_j \in [0,1]$. Now, let us consider three cases, namely $A \notin \theta_x$, $A \in \theta_x$, and $A = \theta_x$,  for $x=1,2,\ldots,l$. Since $m_{i,j}(A) = \overline{m}_{i,j}(A)$ for $A \notin \theta_x$, from \eqref{eqn:bba}, we have $m_{i,j}(A) \in [0,1]$ for $A \notin \theta_x$. Recognizing that $0 \leq \alpha_{i,j}^x \leq 1$, we can conclude from \eqref{eqn:bba} and \eqref{eqn:bba_updated} that $m_{i,j}(A) \in [0,1]$ for $A \in \theta_x$, where $x=1,2,\ldots,l$. For $A=\theta_x$, i.e., $A$ is one of the focal sets, from \eqref{eqn:bba_updated}, we see that 
\begin{equation}
m_{i,j}(A) = \sum_{c_k \in \theta_x} \overline{m}_{i,j}\left(c_k\right) - \max _{c_p \in \theta_x, c_q \in \theta_x, c_p \neq c_q}\left|\overline{m}_{i,j}\left(c_p\right)-\overline{m}_{i,j}\left(c_q\right)\right|, \ \text{for} \ A = \theta_x \ \text{and} \ x= 1,2,\ldots,l. 
\end{equation}
Since $\sum_{c_k \in \theta_x} \overline{m}_{i,j}\left(c_k\right) \le \sum_{c_k \in \overline{\mathcal{F}}} \overline{m}_{i,j}\left(c_k\right)$, where $x = 1,2,\ldots,l$, and $\sum_{c_k \in \overline{\mathcal{F}}} \overline{m}_{i,j}\left(c_k\right) = 1$, we can conclude that $m_{i,j}(A) \in [0,1]$ for $A = \theta_x$.

Now, we will show the normalization property $\sum_{A \in \mathcal{F}} m_{i,j}(A)=1$. Without loss of generality, let us assume that the set $\{c_1,c_2,\ldots,c_q\}$ is partitioned as $\mathcal{S}_1 = \{c_1,c_2,\ldots,c_m\}$ and $\mathcal{S}_2 = \{c_{m+1},\ldots,c_q\}$, where $1 \le m < q$. The set $\mathcal{S}_1$ is such that if $A \in \mathcal{S}_1$, then $A \notin \theta_x$ for $x=1,2,\ldots,l$. For the set $\mathcal{S}_2$, if $A \in \mathcal{S}_2$, then $A \in \theta_x$ for $x=1,2,\ldots,l$. Recall that $\mathcal{F}$ has $(q+l+2)$ elements, summing the belief values for all $A \in \overline{\mathcal{F}}$, we get $\sum_{A \in \mathcal{S}_1} m_{i,j}(A) + \sum_{A \in \mathcal{S}_2} m_{i,j}(A) + m_{i,j}(\emptyset) + m_{i,j}(\Omega)$. Using \eqref{eqn:bba_updated}, this sum can be written as 
\begin{equation}
\sum_{A \in \overline{\mathcal{F}}} m_{i,j}(A) = \sum_{c_k \notin \theta_x} \overline{m}_{i,j}(c_k)
+ \alpha_{i,j}^1 \sum_{c_k \in \theta_1} \overline{m}_{i,j}(c_k)
+ \alpha_{i,j}^2 \sum_{c_k \in \theta_2} \overline{m}_{i,j}(c_k)
+ \ldots
+ \alpha_{i,j}^l \sum_{c_k \in \theta_l} \overline{m}_{i,j}(c_k) + \overline{m}_{i,j}(\Omega). \label{eqn:sum_fbar}
\end{equation}
Using \eqref{eqn:bba_updated} and summing up the belief contributions from the $l$ focal sets, we have 
\begin{equation}
\sum_{A \in \mathcal{F} \backslash \overline{\mathcal{F}}} m_{i,j}(A) = (1-\alpha_{i,j}^1) \sum_{c_k \in \theta_1} \overline{m}_{i,j}(c_k)
+ (1-\alpha_{i,j}^2) \sum_{c_k \in \theta_2} \overline{m}_{i,j}(c_k)
+ \ldots
+ (1-\alpha_{i,j}^l) \sum_{c_k \in \theta_l} \overline{m}_{i,j}(c_k). \label{eqn:sum_focal}
\end{equation}
From \eqref{eqn:sum_fbar} and \eqref{eqn:sum_focal}, we get
\begin{equation}
\sum_{A \in \mathcal{F}} m_{i,j}(A) = \sum_{c_k \notin \theta_x} \overline{m}_{i,j}(c_k)
+ \sum_{c_k \in \theta_1} \overline{m}_{i,j}(c_k)
+ \sum_{c_k \in \theta_2} \overline{m}_{i,j}(c_k)
+ \ldots
+ \sum_{c_k \in \theta_l} \overline{m}_{i,j}(c_k) + \overline{m}_{i,j}(\Omega).
\end{equation}
Recognizing that the right hand side of the preceding equation can be reduced to $\sum_{c_k \in \overline{\mathcal{F}}} \overline{m}_{i,j}\left(c_k\right)$, which equals $1$, we have $\sum_{A \in \mathcal{F}} m_{i,j}(A) = 1$, thus proving the normalization property.
\end{proof}

\subsubsection{Belief Fusion Across Classifiers and Time}
In this block, the beliefs $m_{i,j}$ from the classifiers $\mathbb{C}_j$ are fused to result in a combined belief distribution. Since the backbone of the proposed approach is the DS theory of evidence, it is reasonable to adopt one of the many combination rules proposed in the literature. The Dempster rule of combination, one of the widely adopted combination rules, is found to result in counter-intuitive results when the evidences are highly conflicting \cite{zadeh1986simple}, which is likely to be the case for the aerial target identification and intent prediction problem. To overcome the issue of conflict, the Yager combination rule was proposed in \cite{yager1987dempster}, wherein the conflicting belief is moved to $\Omega$ as shown in \eqref{eqn:yager_combi}. However, Yager’s rule tends to allocate the conflict mass to the frame of discernment $\Omega$, which does not offer useful decision insight. To address this shortcoming, we propose a combination rule that is applicable to belief distributions with focal sets, whereby instead of assigning conflict to \(\Omega\), we redistribute it to specific focal sets \( \theta_1,\theta_2,\ldots,\theta_l\) based on the nature of conflict. The combined belief from the proposed rule provides conclusive and interpretable decisions, especially in scenarios where evidence clearly supports certain groups of targets (focal sets), but remains uncertain between them (elements of the focal sets). Given two belief distributions $m_{i,1}$ and $m_{i,2}$ that are defined over the domain $\mathcal{F}$, the proposed combination rule is defined as
\begin{equation}
\begin{aligned}
m_{i, 1 \cap 2}(A) &= \sum_{B \cap C = A} m_{i,1}(B) m_{i,2}(C), \ \forall A \neq \emptyset \ \text{and} \ A \neq \theta_x, \text{for} \ x = 1,2,\ldots,l  \\
m_{i, 1 \cap 2}(A) &= \sum_{B \cap C = A} m_{i,1}(B) m_{i,2}(C) + K_x, \ \text{for}\  A = \theta_x, \ \text{for} \ x = 1,2,\ldots,l, \ \text{where} \\ 
K_x & = \sum_{\substack{B \cap C = \emptyset \ B \in \theta_x, C \in \theta_x}} m_{i,1}(B) m_{i,2}(C).
\end{aligned}
\label{eqn:generalized_yager}
\end{equation}
Since the preceding belief distribution does not satisfy the property $\sum_{A \in \mathcal{F}} m_{i, 1 \cap 2}(A)=1$, the combined belief is normalized as follows:  
\begin{equation}
m_{i, 1 \cap 2}'(A) = \frac{m_{i, 1 \cap 2}(A)}{\sum_{A \in \mathcal{F}} m_{i, 1 \cap 2}(A)}\ . \label{eqn:comb_norm}
\end{equation}
When no conflict exists within a focal set \(\theta_x\), the corresponding conflict \(K_x\) becomes zero, and the generalized combination rule reduces to the standard Dempster-Shafer rule. In contrast, when significant conflict arises among elements of a predefined focal group, the proposed modification redistributes the conflicting mass within that group. This targeted redistribution preserves class specificity and avoids dilution of belief into the global uncertainty set \(\Omega\). Following this fusion step, the resulting belief distributions corresponding to the $i^{th}$ subsample $X_i$ are combined across the three classifiers to obtain the distribution $m_{i, 1 \cap 2 \cap 3}'$, where \begin{equation}
m_{i, 1 \cap 2 \cap 3}' = m_{i,1} \oplus m_{i,2} \oplus m_{i,3}\ . \label{eq:comb_class}
\end{equation}
In the preceding expression, $\oplus$ denotes the combination operation defined in \eqref{eqn:generalized_yager}. The distribution $m_{i, 1 \cap 2 \cap 3}'$ is then combined with the belief distribution from the $(i-1)^{th}$ subsample to result in the distribution $M_i$, which is the output of the target type classification module for the subsample $X_i$. Specifically, $M_i$ is given by 
\begin{equation}
M_i = m_{i, 1 \cap 2 \cap 3}' \oplus M_{i-1} \label{eq:comb_time}
\end{equation}
where $M_{i-1}$ is the output of the target type classification module for subsample $X_{i-1}$ and is defined over the domain~$\mathcal{F}$.

\subsection{Target Intent Prediction} 
The objective of the target intent prediction module is to infer the intent of an approaching target with respect to the ownship.  In this work, target intent is categorized into three classes: hostile, non-hostile, and suspicious. In addition to the output from the target type classification module, this module incorporates combat geometric features derived from both the target and the ownship. As shown in the right half of Fig.~\ref{fig:flow-diagram}, the module processes the same set of subsamples $X_i$, but focuses on the last $(k-h)$ features in each $x_{i,j}$. Specifically, it uses four relative geometric features: range, attacker deviation angle, off angle, and aspect angle. These features are formally defined in the sequel. A decision tree classifier forms the core of the module. It utilizes handcrafted rules and fuzzy membership degrees to produce a probability distribution across the three intent classes. These probability distributions are then converted into belief distributions using a mapping similar to the one described earlier. Finally, the module applies a belief combination rule--akin to the one employed in the target type classification module--to fuse beliefs across the subsamples. 

In combat scenarios, the true intent of a target may change over time; consequently, the belief combination rule proposed in the previous section will lead to incorrect conclusions due to the combination of conflicting temporal evidence. To overcome this problem, we propose a method to quantify the differences between the evidence using a distance metric introduced in \cite{jousselme2001new}. When the calculated distance remains within a predefined range, belief combination proceeds, thereby refining updates and increasing confidence in intent prediction. If the distance exceeds the threshold, the belief combination is terminated and a new evidence set is initiated to avoid misclassification.

\subsubsection{Features Derived From Combat Geometry}
\label{subsec:combat-geo}
In an aerial combat engagement, the target will strive to position itself in such a manner to gain air superiority. We identify four features that characterize the relative position and orientation of the target with respect to the ownship in an aerial combat scenario \cite{shin2018autonomous}. The aerial combat geometry involving the ownship and a target is depicted in Fig. \ref{fig:combat-geometry}. The four combat geometry features are range, attacker deviation angle, off angle, and aspect angle. 

\begin{figure}[h] 
    \centering
        \includegraphics[width=0.7\linewidth]{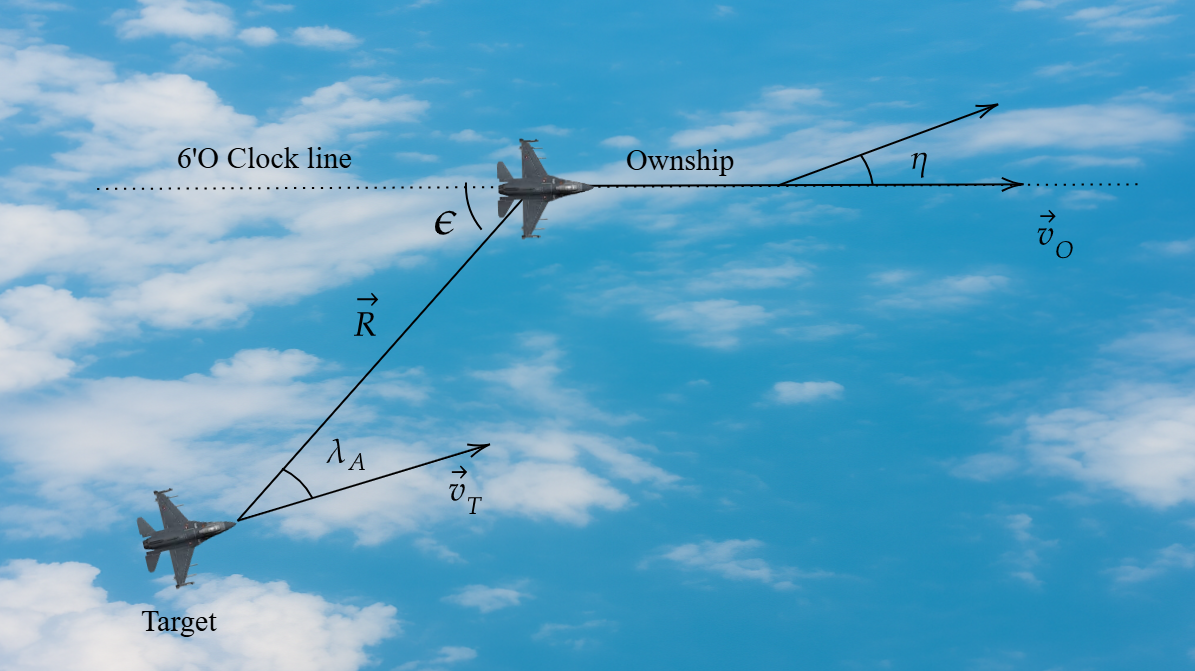}
        \caption{Combat geometry involving the ownship and a target.}
        \label{fig:combat-geometry}
\end{figure}

Range ($r$) represents the distance from the target to the ownship. Aspect angle ($\epsilon$) is the angle between the ownship's 6'O clock line and the line-of-sight (LOS) vector. It provides information about whether the target is ahead or behind the ownship. If $\epsilon<90 \,\textrm{deg}$, the attacker is behind the ownship. If $\epsilon>90\,\textrm{deg}$, the attacker is ahead of the ownship. The attacker deviation angle ($\lambda_A$) is the angle between the target's velocity vector and the LOS vector. It provides information on whether the target's nose is pointed towards the ownship. Off angle ($\eta$) is the angle between the velocity vectors of the target and the ownship. Let $\vec{R}$ represent the position vector of the ownship relative to the target, and $\vec{v}_{O}$ and $\vec{v}_{T}$ denote the velocity vectors of the ownship and target, respectively. While $\vec{v}_{O}$ is obtained from  navigation sensors onboard the ownship, $\vec{R}$ and $\vec{v}_{T}$ are typically estimated by a target tracking module that processes data from radar mounted on the ownship. The mathematical expressions for these four features are provided below: 
\begin{align}
 r &= |\vec{R}|, \quad
\epsilon = \cos^{-1} \left ( \frac{\vec{v}_O \cdot \vec{R}}{|\vec{v}_O| |\vec{R}|} \right ), \quad
\lambda_A = \cos^{-1} \left ( \frac{\vec{v}_T \cdot \vec{R}}{|\vec{v}_T| |\vec{R}|} \right ), \ \text{and} \quad
 \eta = \cos^{-1} \left ( \frac{\vec{v}_T \cdot \vec{v}_O}{|\vec{v}_T| |\vec{v}_O|} \right ) . \label{eqn:combat_features}
\end{align}

\subsubsection{Membership Degree Based on Target Type} \label{ss:membership_degree}
The belief distribution generated by the target-type classification module offers valuable insight into the potential intent of the target. For example, a fighter aircraft is generally more likely to be hostile than a civilian commercial airliner. To quantify the threat level of the target, we aggregate the belief distributions from the target-type classification module over semantically grouped target classes. Inspired from the conversion method proposed in \cite{zhang2022information}, wherein probability distributions are used to derive fuzzy membership degrees for threat classification, herein, we employ belief distributions from the target-type classification module to compute the membership degrees. We categorize the targets into three threat levels: low, medium, and high. A class (or set of classes) in $\mathcal{F}$ is mapped to one of three threat levels by aggregating the belief mass assigned to the relevant focal elements. Let \(\mathcal{C}_1,\mathcal{C}_2,\mathcal{C}_3 \in \mathcal{F}\) denote the predefined subsets of $\mathcal{F}$ corresponding to low, medium, and high threat levels, respectively. The membership values are computed as given below: 
\begin{equation}
\begin{aligned}
\mu_I(M_i) &= \sum_{A \in \mathcal{C}1} M_i(A) &  \quad \text{(Low threat)} \\
\mu_{II}(M_i) &= \sum_{A \in \mathcal{C}2} M_i(A) & \quad \text{(Medium threat)} \\
\mu_{III}(M_i) &= \sum_{A \in \mathcal{C}_3} M_i(A) & \quad \text{(High threat)}
\end{aligned}
\label{eqn:membership_value}
\end{equation}

 \subsubsection{Rules Derived From Combat Geometry Features} \label{ss:rules_combat}
 The decision rules for the decision tree classifier are outlined next. The basic fighter maneuver doctrine offers tactical guidelines pilots use in close-range air combat to achieve a positional advantage over an opponent \cite{datagen}. Based on this doctrine, we derive three rules using the four previously described combat geometry features: range, attacker deviation angle, off angle, and aspect angle. When these three rules are satisfied, they indicate that the target is executing an attacking maneuver relative to the ownship. The three rules are detailed below.
 
\begin{itemize}
\item[] $\mathcal{R}_1$: The distance between the target and the ownship is decreasing. Mathematically, $\frac{d{r}}{dt}<0.$
\item[] $\mathcal{R}_2$: The target is pointing towards the ownship, i.e., $\lambda_A < f(\epsilon)$.
\item[] $\mathcal{R}_3$: The target points toward the ownship when it is behind. Mathematically, $\eta < 90^\circ, \ \text{when} \ \epsilon < 90^\circ$. 
\end{itemize}
where $f(\epsilon) = \begin{cases}   \epsilon,  & \epsilon < 90^0 \\  180^0 - \epsilon, & \epsilon > 90^0 \end{cases}$.
With an objective to reduce the computational demand on the decision tree classifier, we define three additional variables, $R_i$, for $i=1,2,3$, to denote the fraction of time steps within a subsample during which the rule $\mathcal{R}_i$ is satisfied. The fraction $R_i$ is computed as
\begin{equation}
R_i = \frac{1}{w} \sum_{t=1}^{w} \mathbf{1}(\mathcal{R}_i(t))
\label{eqn:rule_extract}
\end{equation}
where \( \mathbf{1}(\cdot) \) is the indicator function that equals one if its argument is true, and zero otherwise. This transformation reduces the dimensionality of the training data from $n \times 3w$ to $n \times 3$. Combined with the three previously defined membership degrees, we now form a feature vector of size $6 \times 1$ consisting of $\mu_I$, $\mu_{II}$, $\mu_{III}$, and $R_i$ for $i=1,2,3$. This feature vector serves as input to the rule-based decision tree classifier. The output probability distribution from the classifier is converted to BBA using a mapping similar to \eqref{eqn:bba}. The resulting belief distribution that corresponds to the $i^{th}$ subsample $X_i$ is denoted as $n_{i}$, where $n_i$ is defined over the domain $\mathcal{G} = \{ \texttt{Hostile}, \texttt{Non-hostile}, \texttt{Suspicious},\texttt{$\Omega$},\ \texttt{$\varnothing$}\}$. 

\subsubsection{Distance-Based Combination Rule} \label{ss:intent_comb_rule}
In dynamic scenarios where target intent evolves rapidly over time, applying the DS combination rule or the Yager rule or the combination rule proposed in \eqref{eqn:generalized_yager} across successive subsamples can be problematic. The primary challenge arises from the fact that conflicting evidence may be amplified through repeated combinations, potentially resulting in unreliable conclusions. This highlights the need for alternative strategies to mitigate the cumulative effects of dynamic inconsistencies in rapidly evolving environments. Let us consider the belief distributions $n_{i}$ and $N_{i-1}$, where $N_{i-1}$ is the output of the target intent prediction module corresponding to subsample $X_{i-1}$. The degree of conflict between the two pieces of evidence can be quantified using the evidence distance \cite{jousselme2001new}, denoted as $d(n_i, N_{i-1})$ and defined in~\eqref{eqn:evidence_distance}. We propose a modification to the standard DS combination rule \eqref{eqn:dempster_combi} by integrating a distance-based stopping criterion. The underlying idea is that if $d(n_i, N_{i-1}) \geq x^*$, where $x^*$ is a predefined threshold, the high level of conflict indicates that the combination process should be restarted from that point onward rather than carried forward. The threshold $x^*$ is determined by solving an optimization problem, which aims to maximize the classification accuracy. 
The modified combination rule is formally defined as
\begin{equation}
n_i \underset{CD}{\oplus} N_{i-1} = \begin{cases} n_i \underset{DS}{\oplus} N_{i-1}, \quad \text{if} \ d(n_i, N_{i-1}) < x^*  \\ n_i,  \quad \text{if} \ d(n_i, N_{i-1}) \geq x^*  \end{cases}. \label{eqn:comb_rule_dist}
\end{equation}
The output of the target intent prediction module for the $i^{th}$ subsample $X_i$ is denoted as $N_i$, where $N_i = n_i \underset{CD}{\oplus} N_{i-1}$. 

Algorithm~\ref{alg:target_classification} provides the pseudocode for the aerial target classification and intent prediction algorithm.

\begin{algorithm}[ht]
\doublespacing
    \caption{Aerial Target Classification and Intent Prediction}\label{alg:target_classification}
    \begin{algorithmic}[1]
    \State \textbf{Input:} A multivariate time series \( X = \{X_1, X_2, X_3, \ldots \} \) consisting of subsamples \( X_i = [x_{i,1}, x_{i,2}, \dots, x_{i,w}] \).
    \State \textbf{Output:} Target type ($M_i$) and intent ($N_i$) belief distributions.
    \Procedure{Target type classification}{$X_i$, $M_{i-1}$}
        \State Extract statistical features $Y_i$ from $X_i$ (Section~\ref{ss:class_features}).
        \State Obtain probability distributions $\mathbb{P}_{i,j}(\cdot)$ for $j=1,2,3$ using trained FFNN, SVM, and LSTM classifiers.        
        \State Convert $\mathbb{P}_{i,j}(\cdot)$ to BBA $m_{i,j}$ using \eqref{eqn:bba_updated}.
        \State Fuse $m_{i,j}$ across classifiers using the combination rule in \eqref{eqn:generalized_yager} to obtain $m_{i, 1 \cap 2 \cap 3}'$ using \eqref{eqn:comb_norm} and \eqref{eq:comb_class}.
        \State Combine evidence across time by fusing $m_{i, 1 \cap 2 \cap 3}'$ with $M_{i-1}$ to get $M_i$ using \eqref{eq:comb_time}.
        \State \Return Target type belief distribution $M_i$.
            \EndProcedure
    \Procedure{Target intent prediction}{$X_i$, $M_i$, $N_{i-1}$}
    \State Extract features from combat geometry using \eqref{eqn:combat_features} (Section~\ref{subsec:combat-geo}).
    \State Convert target type beliefs $M_i$ to membership values $\mu_I$, $\mu_{II}$, $\mu_{III}$ using \eqref{eqn:membership_value} (Section~\ref{ss:membership_degree}).
        \State Derive rules from combat geometry features using \eqref{eqn:rule_extract} (Section~\ref{ss:rules_combat}).
        \State Form feature vector using rules and membership degrees; obtain belief distribution $n_i$ (Section~\ref{ss:rules_combat}).
        \State Fuse evidence across time by fusing $n_i$ and $N_{i-1}$ using the combination rule in \eqref{eqn:comb_rule_dist} (Section~\ref{ss:intent_comb_rule}).
        \State Return target intent belief distribution $N_i$.
    \EndProcedure
    \end{algorithmic}
\end{algorithm}

\section{Synthetic Dataset Generation}
\label{sec:data_generation}
A large and diverse training dataset is key to achieving good performance for the classifiers discussed in the previous section. However, in the context of aerial combat, access to such datasets is often limited and not publicly available due to security concerns. Moreover, collecting experimental data in this domain can be highly impractical. Consequently, there is a need for a simplified method to generate a synthetic dataset containing trajectories of various aerial vehicles operating in a combat environment.

In the domain of aerial target classification, existing studies typically rely on simulated datasets; however, to the best of our knowledge, none of these works \cite{zhang2022information, wang2023quick,zhang2024target,wang2024intelligent,singh2014dynamic} provide details on how the data was generated or the specific conditions under which the simulations were performed. The lack of publicly available datasets prevents reproducibility and poses challenges for further research progress in this domain. To overcome these limitations, we generate our own dataset, enabling precise control over the simulation process and ensuring that it is tailored to the operational scenarios relevant to this study. The dataset generation process has been detailed in our previous work \cite{datagen}.

In \cite{datagen}, we introduced an optimization-based technique that uses basic kinematic target motion models to produce diverse trajectories of various aerial targets operating in a combat environment. In this study, we employ that approach to generate a dataset of aerial target trajectories encompassing multiple target types, as listed in Table \ref{tab:targetTypes}. For completeness, a brief overview of the method from \cite{datagen} is provided before detailing the dataset used to train and evaluate the aerial target classification and intent prediction modules.

\subsection{Optimization-Based Approach For Dataset Generation}
Kinematic motion models are employed to generate the aerial target trajectories, with dynamic feasibility ensured by incorporating target-specific dynamic constraints. We consider eight distinct aerial target types: \texttt{Air-to-air missile}, \ \texttt{Bomber}, \ \texttt{Civilian commercial airliner}, \ \texttt{Civilian helicopter}, \texttt{Military helicopter}, \texttt{Multi-role} \ \texttt{fighter}, \  \texttt{Small civilian aircraft}, and \texttt{Small unmanned aircraft}. For each target type, operational limits on parameters such as altitude, velocity, climb rate, turn rate, and acceleration are defined based on typical mission profiles. These constraints are either provided by domain experts or derived from historical data for the respective target. The specific constraint values for the eight target types are provided in Table \ref{tab:targetTypes}. 

The objective is to generate a dataset of trajectories representing one-versus-one engagements between the ownship and a single target. In a combat scenario, a target may or may not engage in an “attack” on the ownship. Additional engagement-specific constraints, derived from aerial combat geometry analogous to those studied in~\cite{shin2018autonomous}, are formulated to synthesize attacking trajectories. The trajectory generation process is formulated as an optimization problem, wherein the goal is to determine kinematic model parameters and initial positions that minimize constraint violations while satisfying the target-specific constraints. To ensure diversity in the resulting trajectory database, a coverage metric based on cell occupancy is utilized to systematically explore the state space. Inspired by the coverage-guided test generation system proposed in \cite{dang2009coverage}, the metric directs the input signal generation process toward previously unexplored regions of the search space. In this context, the state space comprises velocity, climb rate, and acceleration.
\begin{figure}[h] 
    \begin{subfigure}{0.47\textwidth}
        \centering
        \includegraphics[width=\linewidth]{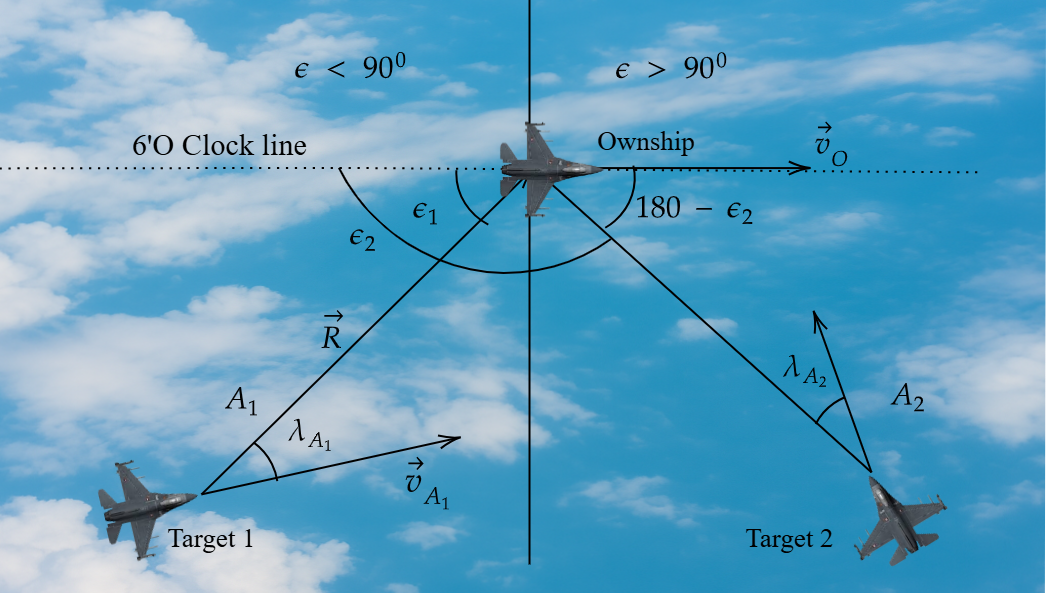}
        \caption{A representative attacking scenario.}
        \label{fig:lambda-attacking}
    \end{subfigure}
    \hfill 
    \begin{subfigure}{0.47\textwidth}
        \centering
        \includegraphics[width=\linewidth]{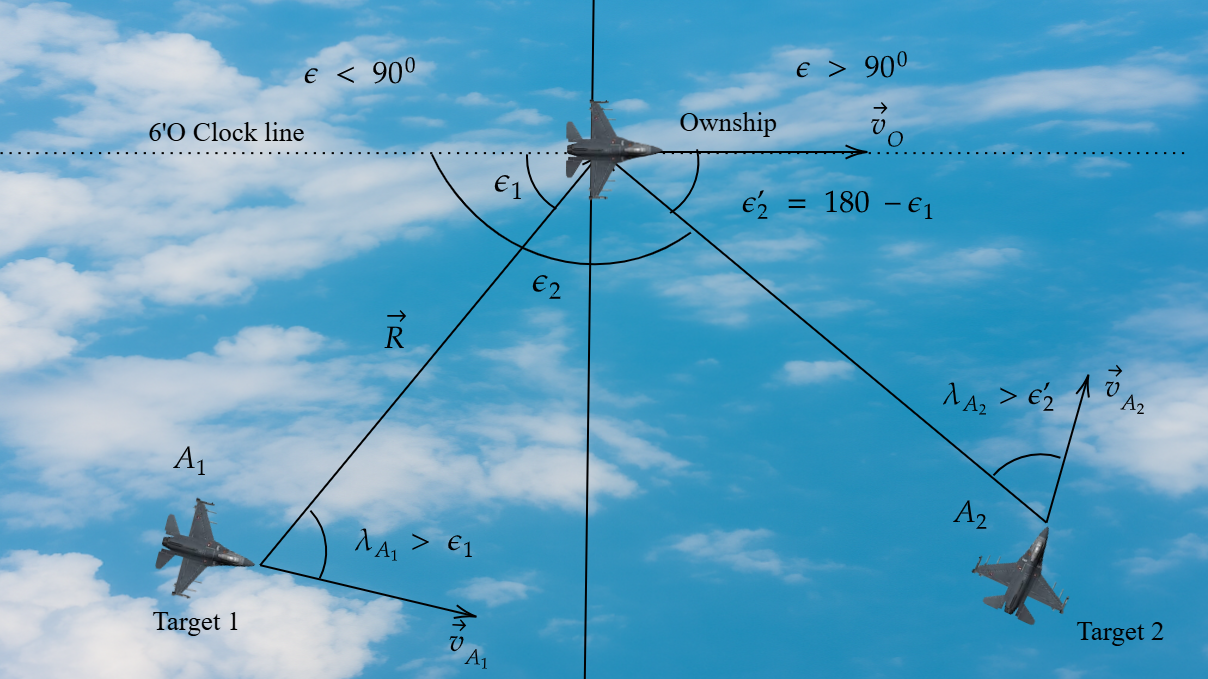}
        \caption{A representative non-attacking scenario.}
        \label{fig:lambda-nonattacking}
    \end{subfigure}
    \caption{Aerial combat scenarios considered in this work.}
    \label{fig:combat-scenarios} 
\end{figure}
In this work, four different motion models — Singer model, Curvilinear model, 3D constant turn model, and Generalized coordinated turn model — are used to generate diverse trajectories by suitably adjusting the model parameters. The aerial combat parameters used to formulate the combat-specific constraints are shown in Fig.~\ref{fig:combat-geometry}. A trajectory segment where the target is pointing towards and closing on the ownship is considered an \texttt{attacking} trajectory; otherwise, the trajectory segment is considered \texttt{non-attacking}. For further details on the specific parameters used in each model and the combat-specific constraints, the interested reader is referred to \cite{datagen}. 
The target trajectories are generated by enforcing target-specific constraints, denoted as $\mathbf{g}(\mathbf{x}(\cdot))$, where $\mathbf{x}(\cdot)$ represents the target state comprising position, velocity, acceleration, and turn rate. The constraints for the eight target types considered in this work are listed in Table~\ref{tab:targetTypes}. In addition to these constraints, the generated trajectories are designed to include both \texttt{attacking} and \texttt{non-attacking} segments. The attacker-specific constraints, expressed as $\mathbf{h}(\mathbf{x}(\cdot), \mathbf{x}_o(\cdot))$, are defined in terms of four combat geometry features—range, attacker deviation angle, off-angle, and aspect angle—where $\mathbf{x}_o(\cdot)$ denotes the ownship state. Both target-specific and attacker constraints are subject to upper and lower bounds, represented by $\mathbf{G}_{UB}, \mathbf{G}_{LB}$, $\mathbf{H}_{UB}, \ \text{and} \ \mathbf{H}_{LB}$, respectively. The primary objective of the proposed trajectory generation algorithm is to identify suitable target motion model parameters and initial target positions such that the resulting trajectory satisfies these constraints with minimal violations. The objective function to minimize the target-specific constraint violations is given by 
\begin{align}
\mathcal{J}_T = \sum_{i=1}^k \lambda_{i} \left(  \sum_{j=1}^N  \mathcal{F}(g_{i}(\mathbf{x}(j)) - G_{UB_{i}}) + \mathcal{F}(G_{LB_{i}} - g_{i}(\mathbf{x}(j)))\right)    
\label{eqn:obj-fn-1}
\end{align} \\
where $k$ is the number of target-specific constraints and $N$ is the maximum number of time samples in the trajectory. For an \texttt{attacking} trajectory segment, the following objective function is minimized within an attack window $w_a = n_b - n_a$: 
    \begin{align}
    \mathcal{J}_A = \sum_{i=1}^m \lambda_{i}' \left(\sum_{j={n_a}}^{n_b} \mathcal{F}(h_{i}(\mathbf{x}(j)) - H_{UB_{i}}) +  \mathcal{F}(H_{LB_{i}} - h_{i}(\mathbf{x}(j))) \right)
    \label{eqn:obj-fn-2}
    \end{align}
where $m$ is the number of attacker constraints. $\mathcal{F}(x)$ is an indicator function used to quantify constraint violations and is defined as $\mathcal{F}(x) = 0, \ \text{if} \ x < 0$ and $\mathcal{F}(x) = x, \ \text{if} \ x \ge 0$. $\lambda \text{ and } \lambda'$ represent the weights associated with each constraint.
To generate an attacking trajectory, the sum of both the objective functions given by \eqref{eqn:obj-fn-1} and \eqref{eqn:obj-fn-2} is minimized. Thus, the objective function for the optimization problem is given by
\begin{equation}
    \mathcal{J} = \mathcal{J}_T + \mathcal{J}_A.
\end{equation}
\begin{table}[htbp]
    \caption{Table showing the various targets considered in this work and their associated constraints used for dataset generation.}
    \centering
    \begin{tabular}{lp{7em}p{5em}p{5em}p{5em}p{5em}}
    \hline \hline
    \textbf{Target Type} & \textbf{Altitude (ft)} & \textbf{Velocity (m/s)} & \textbf{Climb Rate (m/s)} & \textbf{Turn Rate (rad/s)} & \textbf{Acceleration (g)}  \\
    \hline 
    Air-to-air missile & 100 to 100,000 & 200 to 500 &  100 to 200 & 0.3 to 0.4 & 1 to 10 \\
    Bomber   & 0 to 50,000 & 75 to 105 &  2 to 25 & 0.05 to 0.2 & -1 to 4 \\
    Civilian commercial airliner & 0 to 45,000 & 75 to 105 &  2 to 40 & 0.05 to 0.2 & 0 to 3 \\
    Civilian helicopter   & 0 to 15,000 & 0 to 200 &  2 to 20 & 0.05 to 0.1 & 0 to 2 \\ 
    Military helicopter   & 0 to 20,000 & 0 to 250 &  2 to 28 & 0.05 to 0.1 & 0 to 3 \\
    Multi-role fighter   & 0 to 50,000 & 50 to 390 &  2 to 110 & 0.05 to 0.4 & -3 to 9\\
   Small civilian aircraft & 0 to 36,000 & 75 to 105 &  2 to 35 & 0.05 to 0.2 & 0 to 3 \\
   Small unmanned aircraft & 0 to 10,000 & 20 to 50 &  2 to 10 & 0.05 to 0.1 & 0.5 to 4 \\
    \hline \hline
    \end{tabular}
    \label{tab:targetTypes}
\end{table}

In addition to features derived from target kinematics, the target RCS is also used as an input feature for target type classification. RCS is the measure of the incident electromagnetic energy reflected or scattered back towards the radar source. Several software tools are available for computing the RCS of complex-shaped objects, as discussed in \cite{sehgal2019automatic}. In this work, we employ the \textit{MGL-FASTRCS} tool \cite{cakir2008radar,cakir2014fdtd}, chosen for its computational efficiency and parallelization capabilities, to estimate the RCS of various aerial targets under consideration.


\section{Experimental Analysis}
\label{sec::experimental_evaluation}
This section presents a comprehensive analysis of the proposed integrated target type classification and intent prediction framework. 
As discussed in Section~\ref{sec:related}, since there is no common, publicly available benchmark dataset for military target intention recognition, and each study uses highly specialized custom datasets with unique features and labeling conventions, it is not feasible to directly compare our approach with existing techniques. Therefore, performance of the framework is evaluated against baseline approaches using key metrics such as classification accuracy, intent prediction accuracy, effectiveness of belief combination, and computational efficiency. The evaluation is carried out on the aerial target dataset described in the preceding section, which consists of time-series data representing one-to-one engagements between the ownship and an aerial target in a combat scenario. All classifiers were trained and tested on a system running Ubuntu 22.04.3 LTS, equipped with a 11th Gen Intel® Core™ i9-11900 CPU @ 2.50 GHz and 31 GB of RAM. 

\subsection{Input Dataset}
The aerial target dataset comprises $500$ trajectories for each target class. Each trajectory is sampled at \(50~\text{Hz}\) over a duration of $20$ seconds, resulting in $1000$ time steps per trajectory. Every trajectory is characterized by $13$ distinct features: velocity components \((u_x, u_y, u_z)\), acceleration components \((a_x, a_y, a_z)\), altitude \((z)\), turn rate \((\dot{\phi})\), radar cross-section (RCS), range \((r)\), attacker deviation angle \((\lambda_A)\), off-angle \((\eta)\), and aspect angle \((\epsilon)\). These features are either directly obtained from onboard sensors such as accelerometers, gyroscopes, and radar, or derived from such measurements. Collectively, they provide strong discriminative power by capturing the physical and dynamic characteristics of different targets. For ease of reference, the eight target classes are labeled as follows: \texttt{Bomber} - $c_1$, \texttt{Civilian commercial airliner} - $c_2$, \texttt{Civilian helicopter} - $c_3$, \texttt{Air-to-air missile} - $c_4$, \texttt{Military helicopter} - $c_5$, \texttt{Multi-role fighter} - $c_6$, \texttt{Small civilian aircraft} - $c_7$, and \texttt{Small unmanned aircraft} - $c_8$. In addition to the target class, each trajectory is also characterized by its engagement state, which can be either \texttt{attacking} or \texttt{non-attacking}. It is emphasized that, although the true target class remains invariant over the entire trajectory, the engagement state is dynamic and may transition between \texttt{attacking} and \texttt{non-attacking} modes within a trajectory. The generated trajectory dataset is partitioned into training set ($70\%$), testing set ($20\%$), and validation set ($10\%$). Each trajectory, consisting of $1000$ time steps with $13$ distinct features per step, is divided into non-overlapping subsamples of length $50$, resulting in $20$ subsamples per trajectory. 
\subsection{Target Type Classification} \label{ss:target_type_results}
The target type classification module uses only a subset of the $13$ features, namely \(u_x, u_y, u_z\), \(a_x, a_y, a_z\), \(z\), \(\dot{\phi}\), and $RCS$. Statistical features are extracted from each subsample as mentioned in section~\ref{ss:class_features}. After feature extraction, the resulting representations are used to train the three classifiers used in the target type classification module. FFNN and SVM are trained on the derived statistical features, while the LSTM classifier directly processes raw subsamples. The optimal hyper-parameters for each classifier are summarized in Table~\ref{tab:hyperparameter}. To evaluate each classifier's performance, the average classification accuracy on the test set is computed using the following metric:
\begin{equation}
Accuracy = \frac{Number \ of \ Subsamples \ Correctly \ Predicted}{Total\  Number\  of \ Subsamples}.
\label{eqn:subsample_accuracy}
\end{equation}

\begin{table}[ht]
\centering
\caption{Optimized hyper-parameters for classifiers used in the target type classification module.}
\label{tab:hyperparameter}
\begin{tabular}{lll} 
\hline \hline
\textbf{Classifier} & \textbf{Hyperparameter} & \textbf{Optimal Value} \\
\hline
\multirow{2}{*}{SVM} 
  & Regularization Parameter ($C$) & 2  \\ 
  & Kernel Function & RBF   \\ 
\hline
\multirow{7}{*}{FFNN} 
  & Activation Function & Sigmoid    \\ 
  & Hidden Layer 1 Size & 99     \\ 
  & Hidden Layer 2 Size & 99     \\ 
  & Optimizer & Adam     \\ 
  & Dropout Rate & 0.2 \\
  & Maximum Iterations & 300    \\ 
  & Early Stopping & Enabled  \\
\hline
\multirow{4}{*}{LSTM} 
  & LSTM Layer 1 Size & 9 \\ 
  & LSTM Layer 2 Size & 9 \\
  & Maximum Epochs & 300 \\
  & Early Stopping Patience & 5  \\
\hline \hline
\end{tabular}
\end{table}

Based on the metric defined in \eqref{eqn:subsample_accuracy}, the FFNN classifier achieves the highest accuracy of \(81\%\), followed by the SVM classifier at \(76\%\), and the LSTM classifier at \(74\%\) on the testing set. These classifiers generate a probability distribution over the set of eight target classes $c_i$ for $i=1,2,\ldots,8$. The probability distribution from classifier $\mathbb{C}_j$ for subsample $i$ is mapped to a belief distribution using \eqref{eqn:bba} resulting in $\overline{m}_{i,j}$, which is a $10 \times 1$ vector given by $\{\overline{m}_{i,j}(c_1), \overline{m}_{i,j}(c_2), \ldots, \overline{m}_{i,j}(c_8), \overline{m}_{i,j}(\Omega), \overline{m}_{i,j}(\emptyset)\}$. It is found that targets such as \texttt{Bomber} and \texttt{Civilian commercial airliner} exhibit low classification accuracy as they possess closely aligned kinematic features and RCS values, as shown in Table~\ref{tab:targetTypes}. A similar observation is made for the target group consisting of \texttt{Civilian helicopter} and \texttt{Military helicopter}. The similarity in the features presents significant challenges and to overcome this issue, we create additional focal sets $\theta_1=\left\{c_3, c_5\right\}$ and $\theta_2=\left\{c_1, c_2\right\}$ to represent uncertainty within these two target groups as discussed in Section~\ref{ss:class_mapping}. The BBA with the new focal elements is created using \eqref{eqn:bba_updated}, which results in $\{m_{i,j}(c_1), m_{i,j}(c_2), \ldots, m_{i,j}(c_8), m_{i,j}(\theta_1),m_{i,j}(\theta_2),m_{i,j}(\Omega), m_{i,j}(\emptyset)\}$. The resulting belief distributions are combined across multiple classifiers and over subsamples using \eqref{eqn:generalized_yager}. Now that the classifiers are trained, we evaluate the performance of the target type classification module over entire trajectories. We employ additional trajectory-level evaluation criteria beyond standard accuracy metrics. Specifically, we define the following two criteria to quantify the average classification accuracy for a trajectory: 

\noindent \textit{(a) Non-Consecutive (NC) criterion} - A trajectory is said to be correctly classified if the target type classification module in Algorithm~1 makes a correct prediction in more than ten subsamples, regardless of whether these correctly predicted subsamples are consecutive. Since the output of the target classification algorithm is a belief distribution, a prediction is deemed correct if the class or focal set pertaining to the maximum belief value contains the true target class. The average accuracy with the NC criterion is then defined as 
\begin{equation}
Accuracy_{NC} = \frac{\beta}{n} \times 100\%
\end{equation}
where $\beta$ is the number of trajectories in which more than ten non-consecutive windows are correctly predicted, and $n$ is the total number of trajectories tested.  

\noindent \textit{(b) Consecutive (C) criterion} - If the algorithm makes a correct prediction in more than ten consecutive subsamples, the trajectory is said to be correctly classified. The resulting average accuracy is defined as 
\begin{equation}
Accuracy_{C} = \frac{\gamma}{n} \times 100\%
\end{equation}
where $\gamma$ is the number of trajectories in which more than ten consecutive windows are correctly predicted. 

\begin{table}[ht]
    \centering
    \caption{Average accuracies of different approaches on the aerial target dataset.}
    \begin{tabular}{lccc}
        \hline \hline
        \textbf{Method} & \textbf{Accuracy} & $\boldsymbol{\mathrm{Accuracy}_{C}}$ & $\boldsymbol{\mathrm{Accuracy}_{NC}}$ \\
        \hline
        FFNN & 0.81  & 0.73 & 0.82  \\
        LSTM & 0.74
        
        & 0.68 & 0.78 \\
        SVM & 0.76 & 0.65 & 0.77 \\
        $DS_{LSTM  \,\bigoplus\, SVM}$ & 0.79 & 0.70 & 0.81  \\
        $DS_{FFNN  \,\bigoplus\, LSTM}$ & 0.81 & 0.75 & 0.82  \\
        $DS_{FFNN  \,\bigoplus\, SVM}$ & 0.81 & 0.72 & 0.82  \\ 
         $DS_{FFNN \,\bigoplus\ SVM \bigoplus LSTM}$ & 0.84 & 0.81 & 0.84 \\
       $DS_{FFNN \,\bigoplus\, SVM}^{(1:t)}$ & 0.84 & 0.83 & 0.84  \\
        $DS_{FFNN \,\bigoplus\, LSTM}^{(1:t)}$ & 0.83 & 0.82 & 0.84  \\
        $DS_{SVM \,\bigoplus\, LSTM}^{(1:t)}$ & 0.83  & 0.82 & 0.83   \\
        $DS_{FFNN \,\bigoplus\ SVM \bigoplus LSTM }^{(1:t)}$ & 0.85 & 0.83 & 0.85 \\
        Proposed Approach &\bf{0.88} & 0.85 & \bf{0.89} \\
        \hline \hline
    \end{tabular}
    \label{tab:Comparison}
    \vspace{5mm}
\end{table}

\begin{figure}[ht]
    \centering
    \begin{subfigure}{0.49\textwidth}
        \centering        \includegraphics[width=\linewidth]{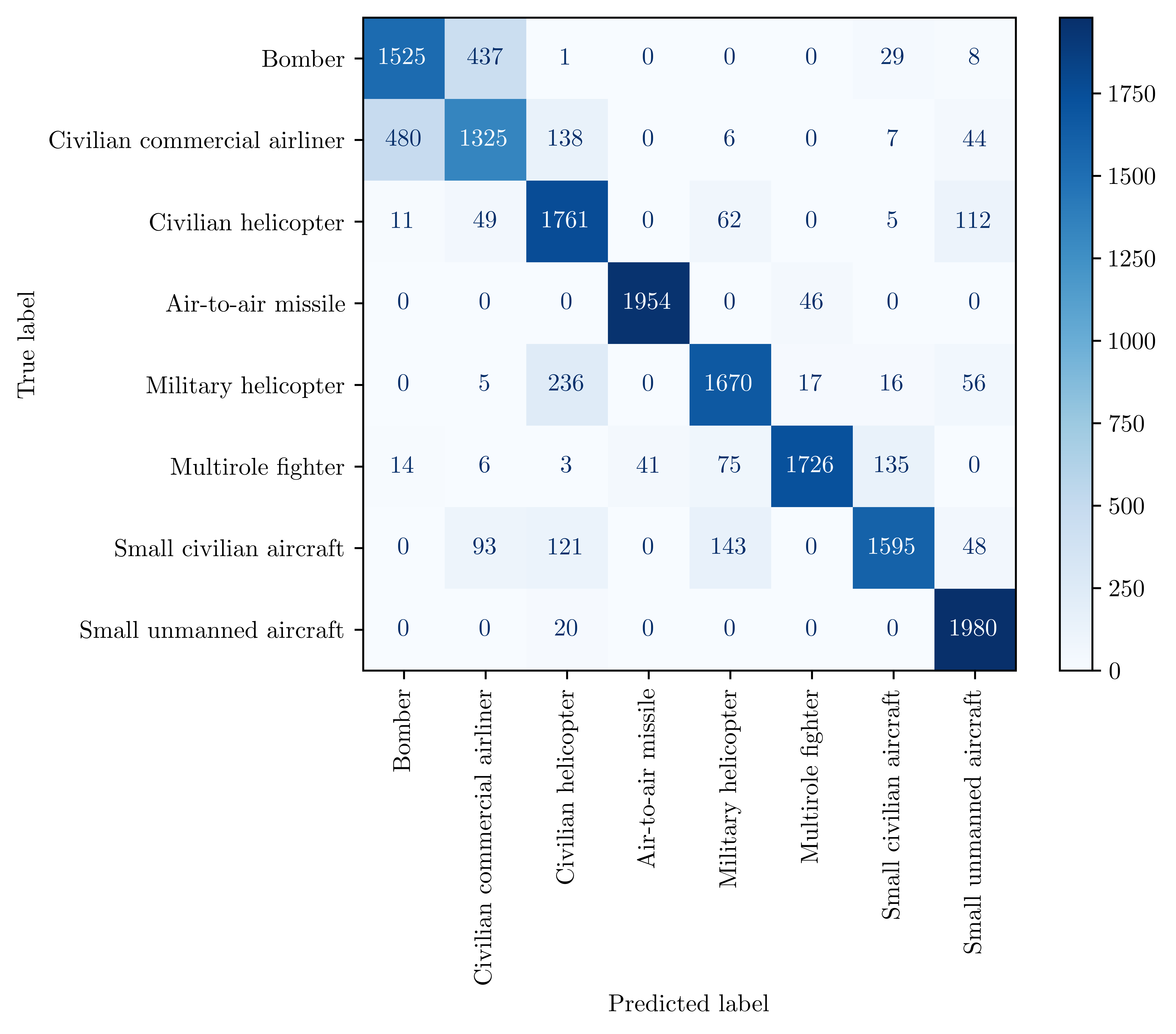}
        \caption{Baseline DS Method ($DS_{FFNN \bigoplus SVM \bigoplus LSTM}$)}
    \end{subfigure}
    \hfill
    \begin{subfigure}{0.49\textwidth}
        \centering
        \includegraphics[width=\linewidth]{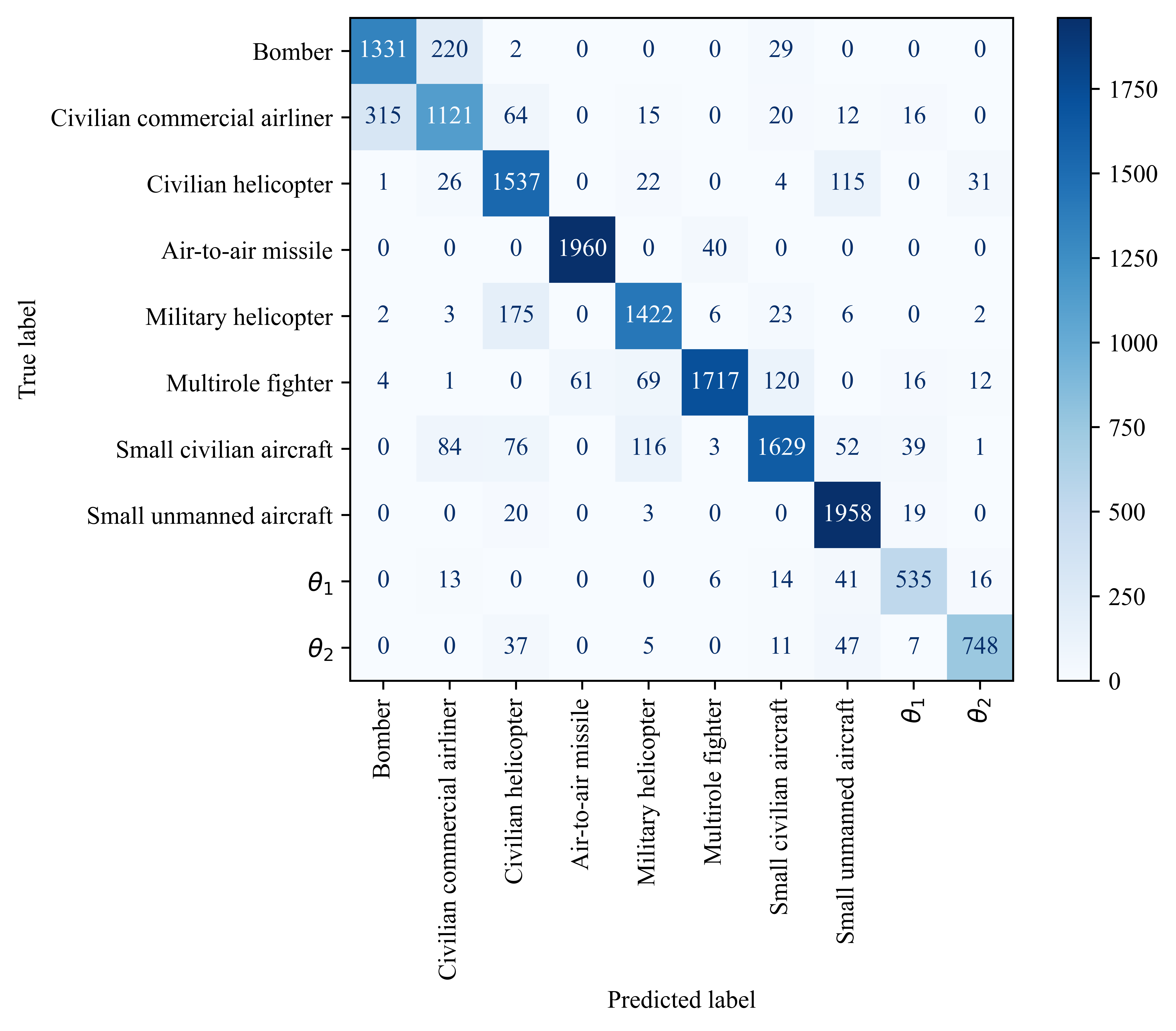}
        \caption{Proposed Approach}
    \end{subfigure}
    \caption{Confusion matrices resulting from evaluations using the proposed approach and the baseline DS method.}
    \label{fig:Confusion_Matrix_Comparison}
\end{figure}

With both the evaluation criteria defined, we now compare the performance of our proposed framework with a baseline method that employs the DS combination rule. In this baseline, the evidences are fused using the DS combination rule defined in \eqref{eqn:dempster_combi}, instead of the proposed rule in \eqref{eqn:generalized_yager}. The results, as shown in Table~\ref{tab:Comparison} and Fig.~\ref{fig:early_prediction}, indicate that our approach outperforms the baseline in several key areas.
\subsubsection{Ablation Study}
In Table~\ref{tab:Comparison}, the second column corresponds to the accuracy metric defined in (26), evaluated over subsamples following belief fusion. The ablation study reveals that individual classifiers (FFNN, LSTM, SVM) provide only moderate performance when used in isolation. However, progressively introducing DS-based fusion yields consistent improvements across all accuracy measures, demonstrating the complementary nature of the classifiers. The proposed approach achieves the highest average accuracy over the test trajectories, surpassing all individual and pairwise fusion baselines. With respect to standard DS fusion, the gain under the standard accuracy metric appears marginal; the improvements become substantially more pronounced under the NC criterion, where robustness to non-category samples is critical. In contrast, performance under the more restrictive C criterion remains comparable to the DS fusion, indicating that the primary benefits of the proposed architecture are realized in more challenging, ambiguity-driven scenarios. Second, a major benefit of our approach is the reduction in false predictions, which is clearly illustrated in the confusion matrix shown in Fig.~\ref{fig:Confusion_Matrix_Comparison}. For instance, the number of misclassified subsamples for the trajectories pertaining to the bomber decreases from $475$ to $251$, highlighting improved reliability in target classification with the proposed~approach. 

We also assessed the early prediction capability of the proposed framework by looking at the subsample starting from which Algorithm 1 correctly predicts the target class. Figure~\ref{fig:early_prediction} depicts a bar graph, where the $x$-axis represents different targets and the $y$-axis indicates the number of trajectories that are correctly classified under the C criterion starting from a specific subsample. The results are particularly promising for \texttt{Air-to-air missile} and \texttt{Small unmanned aircraft}, where more than \(95\%\) of trajectories are correctly classified from the first subsample itself. For targets such as \texttt{Bomber}, \texttt{Civilian helicopter}, \texttt{Multi-role fighter}, and \texttt{Small civilian aircraft}, over \(75\%\) of trajectories are correctly classified from the first subsample. For other cases, consistent correct classification is observed somewhere between the \(2^{nd}\) and \(10^{th}\) subsample. These findings demonstrate the capability of the proposed approach in achieving early and stable predictions. Consequently, it can provide reliable predictions without the need to analyze the complete trajectory, a key advantage in aerial combat scenarios where timely threat assessment and rapid response are critical. In the subsequent sections, we show that as time progresses (i.e., with more subsamples), the proposed approach yields predictions with increasingly higher belief values, and therefore, greater confidence.

\begin{figure}
    \centering
    \includegraphics[width=0.8\textwidth]{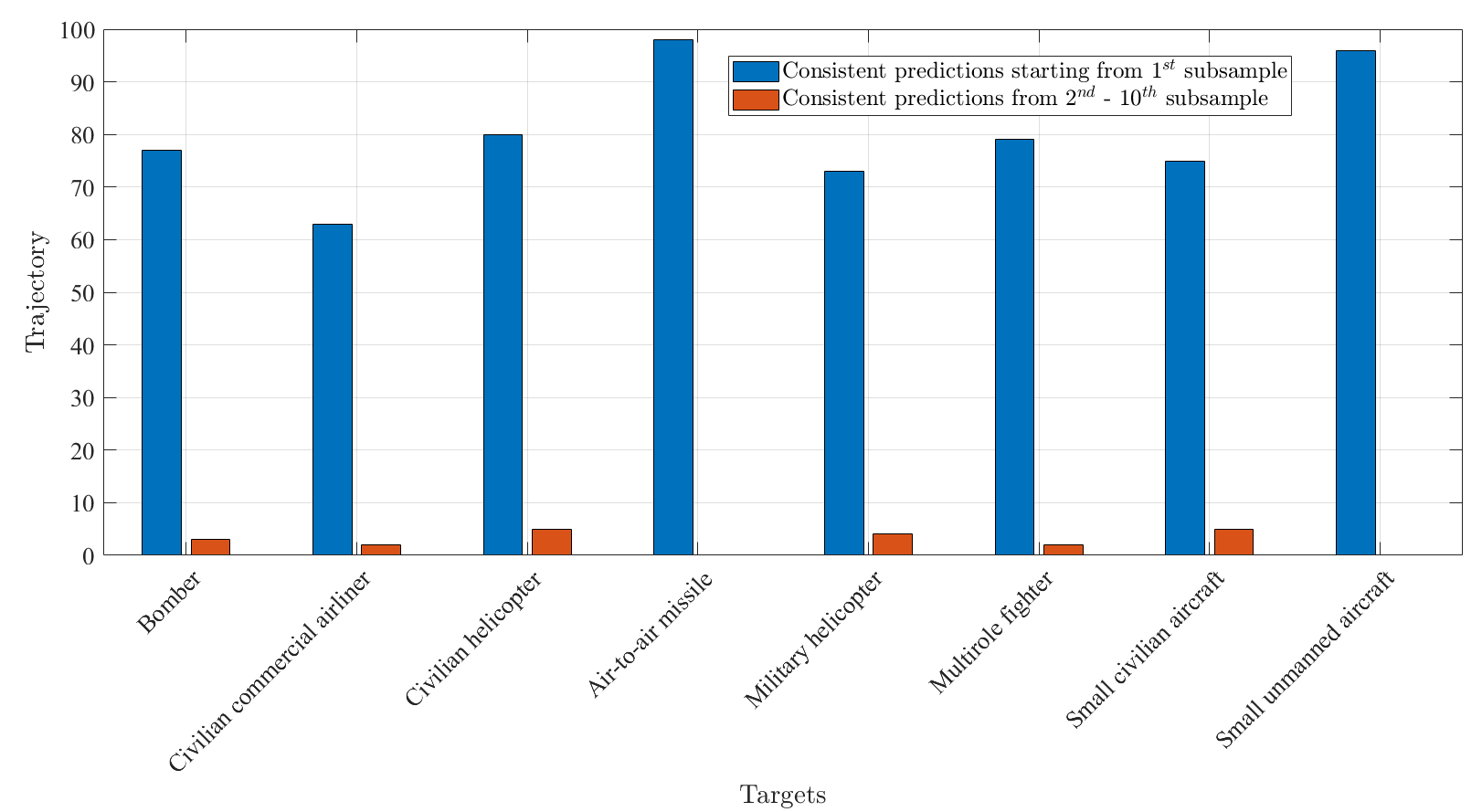}
    \caption{The bar graph shows the number of trajectories with consistent predictions starting from $1^{st}$ and $(2^{nd}-10^{th})$ subsample.}
    \label{fig:early_prediction}
\end{figure}

\subsection{Target Intent Prediction}
The target intent prediction module uses the four features--\(r\), \(\lambda_A\), \(\eta\), and \(\epsilon\)--along with the belief distribution generated by the target type classification module. As discussed in Section~\ref{subsec:combat-geo}, these features are easily extracted from the combat geometry using  \eqref{eqn:combat_features}. In the aerial target dataset, the subsamples are labeled as either \texttt{attacking} or \texttt{non-attacking} using the rules mentioned in Section~\ref{ss:rules_combat}. With the availability of target type belief distribution, we categorize each target into one of three threat levels as shown in Table~\ref{tab:target_levels}. The threat levels are determined based on the target’s known weapon capabilities, operational roles, and strategic intent. With the threat levels defined, the belief output from the target type classification module is converted into a membership value using~\eqref{eqn:membership_value}. 
\begin{table}
\centering
\caption{The main parameters of the three target type levels.}
\label{tab:target_levels}
\begin{tabular}{llll}
\hline \hline
\textbf{Level} & \textbf{Target Type} & \textbf{Weapons} & \textbf{Goals} \\
\hline
\multirow{3}{*}{I} 
& Bomber & Missiles, Bombs & \multirow{3}{*}{Strategic \& Tactical} \\
& Air-to-air missile & Guns & \\
& Multi-role fighter & Electronic Warfare & \\
\hline
\multirow{4}{*}{II} 
& Small unmanned aircraft &  & Reconnaissance \\
& Military helicopter & Drop munitions & Surveillance \\
& $\theta_1$, $\theta_2$ & Gun & \\
& $\Omega$ &  & Target Acquisition \\
\hline
\multirow{3}{*}{III} 
& Civilian commercial airliner & \multirow{3}{*}{No Weapon} & Goods \& Civilian Transport \\
& Civilian helicopter &  & Recreational \& \\
& Small civilian aircraft &  & Surveillance \\
\hline \hline
\end{tabular}
\end{table}
Besides the three-level membership degree, three rule-based features are defined using the combat geometry features (see Section~\ref{ss:rules_combat}), resulting in a feature vector of size $6 \times 1$. Based on information about the target type and the behavioral labels, i.e., \texttt{attacking} or \texttt{non-attacking}, we define the following three intent labels:
\begin{itemize}
    \item[(1)] \textit{Hostile}: Target behavior is \texttt{attacking} and it belongs to either Level I or Level II. Example: If the target is \texttt{attacking} and is either a \texttt{Multi-role fighter} or a \texttt{Military helicopter}, the possibility of the target being hostile towards the ownship is high.
    \item[(2)] \textit{Non Hostile}: Target behavior is \texttt{non-attacking} and it belongs to Level II or Level III. Example: If the target is a \texttt{Civilian commercial airliner} or \texttt{Military helicopter}, and it is \texttt{non-attacking}, the target should be classified as non-hostile.
    \item[(3)] \textit{Suspicious}: A target is classified as suspicious if it meets either of the following conditions: it exhibits \texttt{attacking} behavior and belongs to either Level II or Level III; or it exhibits \texttt{non-attacking} behavior and belongs to Level I. Example: A \texttt{Civilian commercial airliner} displaying \texttt{attacking} behavior is considered suspicious.  Although the target is a Level III entity, its behavior strongly suggests a potential threat.
\end{itemize}
Using the six element feature vector (three membership values and the three rule-base features), a rule-based decision tree model is trained, and an accuracy of $93\%$ is achieved for the testing set. The resulting probability distribution from the classifier is converted to a belief distribution using \eqref{eqn:bba}. Subsequently, the belief distribution is combined across subsamples using \eqref{eqn:comb_rule_dist} (see Section~\ref{ss:intent_comb_rule}). To evaluate the effectiveness of the intent prediction module, we compare the performance of the following three approaches: a baseline decision tree classifier, the conventional DS combination method, and our proposed distance-based adaptive fusion method.
For comparison, the belief is converted back to a probability distribution using \(betp\), which is defined in \eqref{eqn:betp}. Fig.~\ref{fig:proposed_compare} presents a comparative evaluation of class probability distributions across a sequence of subsamples for the three approaches. Each subplot illustrates how these distributions evolve over the subsamples, with the horizontal axis indicating the subsample index and the vertical axis representing the predicted probabilities. The ground-truth label for each subsample is shown in the bottom row using distinct markers. Notably, the conventional DS framework begins to misclassify as early as the fifth subsample. Its rigid fusion rule, which is sensitive to highly conflicting evidence, causes a gradual deterioration in classification performance. In contrast, our proposed framework introduces a novel distance-based adaptive belief fusion mechanism. By dynamically adjusting the fusion process based on evolving belief divergence, it significantly improves the robustness of intent inference over time. This approach not only reduces misclassification but also maintains high predictive confidence across subsamples, as evidenced in Table~\ref{fig:dt_accuracy}. Although the baseline decision tree classifier achieves a high initial accuracy of $93\%$, its static nature limits performance consistency. Or in other words, even when the decision tree classifier makes the correct prediction, the confidence in the prediction is low compared to the proposed approach, as observed from the first subplot in Fig.~\ref{fig:proposed_compare}. The conventional DS method achieves only $85\%$ accuracy due to overcommitment to early evidence. In contrast, our method maintains both high accuracy $93\%$ and resilience to temporal misclassifications, demonstrating the advantage of adaptive belief combination in dynamic combat scenarios.

\begin{figure}[h]
    \centering
    \includegraphics[width=0.8\textwidth]{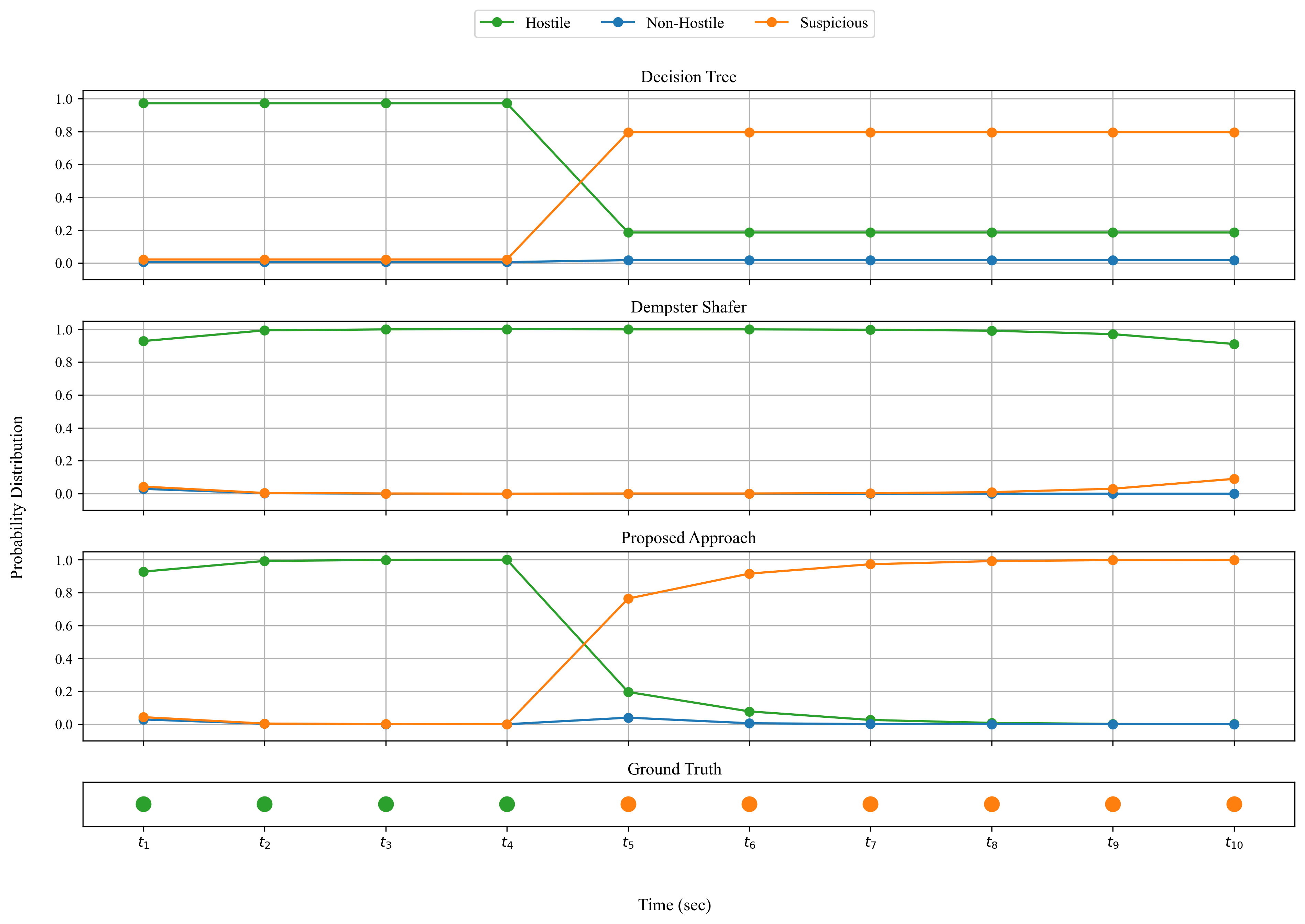}
    \caption{Comparison of outputs from the Decision tree, DS combination method, and the proposed method across successive subsamples.}
    \label{fig:proposed_compare}
\end{figure}
\begin{table}[h]
\centering
\caption{Performance metrics for intent prediction.}
\begin{tabular}{lc} 
\hline \hline
 \textbf{Methods} & \textbf{ Accuracy} \\
\hline
Decision Tree  & 0.93    \\
$DS_{Decision \ Tree} ^{1:t}$ & 0.85  \\
Proposed Approach & \underline{0.93} \\
\hline \hline
\end{tabular}
\label{fig:dt_accuracy}
\end{table}
\subsubsection{Sensitivity Analysis of the Proposed Approach}
A detailed sensitivity analysis is performed in which the test data are corrupted with either additive Gaussian noise or multiplicative noise prior to prediction. Let $x \in \mathbb{R}^d$ denote a clean, scaled input feature vector. For additive Gaussian noise, the perturbed input is defined as
\begin{equation}
\tilde{x} = x + n, 
\qquad 
n \sim \mathcal{N}\!\left(0,\, \sigma^2 I\right),
\label{eq:add_noise}
\end{equation}
where $\sigma \in [0,1]$ denotes the perturbation level and $I$ is the identity matrix, indicating that independent Gaussian perturbations are applied to each feature. For multiplicative noise, the corrupted feature vector is given by
\begin{equation}
\tilde{x} = x \odot (1 + \epsilon), 
\qquad 
\epsilon \sim \mathcal{N}\!\left(0,\, \sigma^2 I\right),
\label{eq:mult_noise}
\end{equation}
where $\odot$ denotes element-wise multiplication. In this case, $\eta$ again controls the normalized noise intensity, inducing proportional distortions in each feature dimension. Table~\ref{tab:noise} reports the classification performance under progressively increasing values of $\sigma$. The results indicate that the proposed approach maintains high predictive accuracy for small and moderate $\sigma$, with a gradual and interpretable degradation as $\sigma$ increases. When $\sigma \geq 0.6$, the magnitude of the perturbations becomes comparable to the scaled feature values, resulting in a pronounced decline in accuracy. These high-noise cases thus characterize the robustness limits of the method, whereas smaller values of $\sigma$ correspond to operationally realistic sensing conditions. 

\begin{table*}[h]
\centering
\caption{Sensitivity analysis of target-type classification and intent prediction under different noise types and perturbation levels (\%).}
\begin{tabular}{c c c c}
\hline
\hline
\textbf{Noise Type} & 
\textbf{Perturbation Level} & 
\textbf{Accuracy\(_{\text{Type Classification}}\)} & 
\textbf{Accuracy\(_{\text{Intent Prediction}}\)} \\
\hline
\hline
Gaussian      & 0.00 & 87.49 & 93.16 \\
              & 0.10 & 87.19 & 92.66 \\
              & 0.20 & 86.79 & 89.88 \\
              & 0.30 & 85.87 & 87.24 \\
              & 0.40 & 83.88 & 83.35 \\
              & 0.50 & 82.05 & 80.98 \\
              & 0.60 & 77.77 & 77.62 \\
              & 0.70 & 72.56 & 74.73 \\
\hline
Multiplicative & 0.00 & 87.49 & 93.16 \\
               & 0.10 & 87.41 & 92.63 \\
               & 0.20 & 87.36 & 90.64 \\
               & 0.30 & 87.18 & 88.41 \\
               & 0.40 & 86.44 & 85.81 \\
               & 0.50 & 86.58 & 84.75 \\
               & 0.60 & 84.10 & 83.33 \\
               & 0.70 & 83.43 & 83.00 \\
\hline
\hline
\end{tabular}
\label{tab:noise}
\end{table*}

\begin{figure}[h]
    \centering
    \begin{subfigure}{0.32\textwidth}
        \centering
        \includegraphics[width=\linewidth]{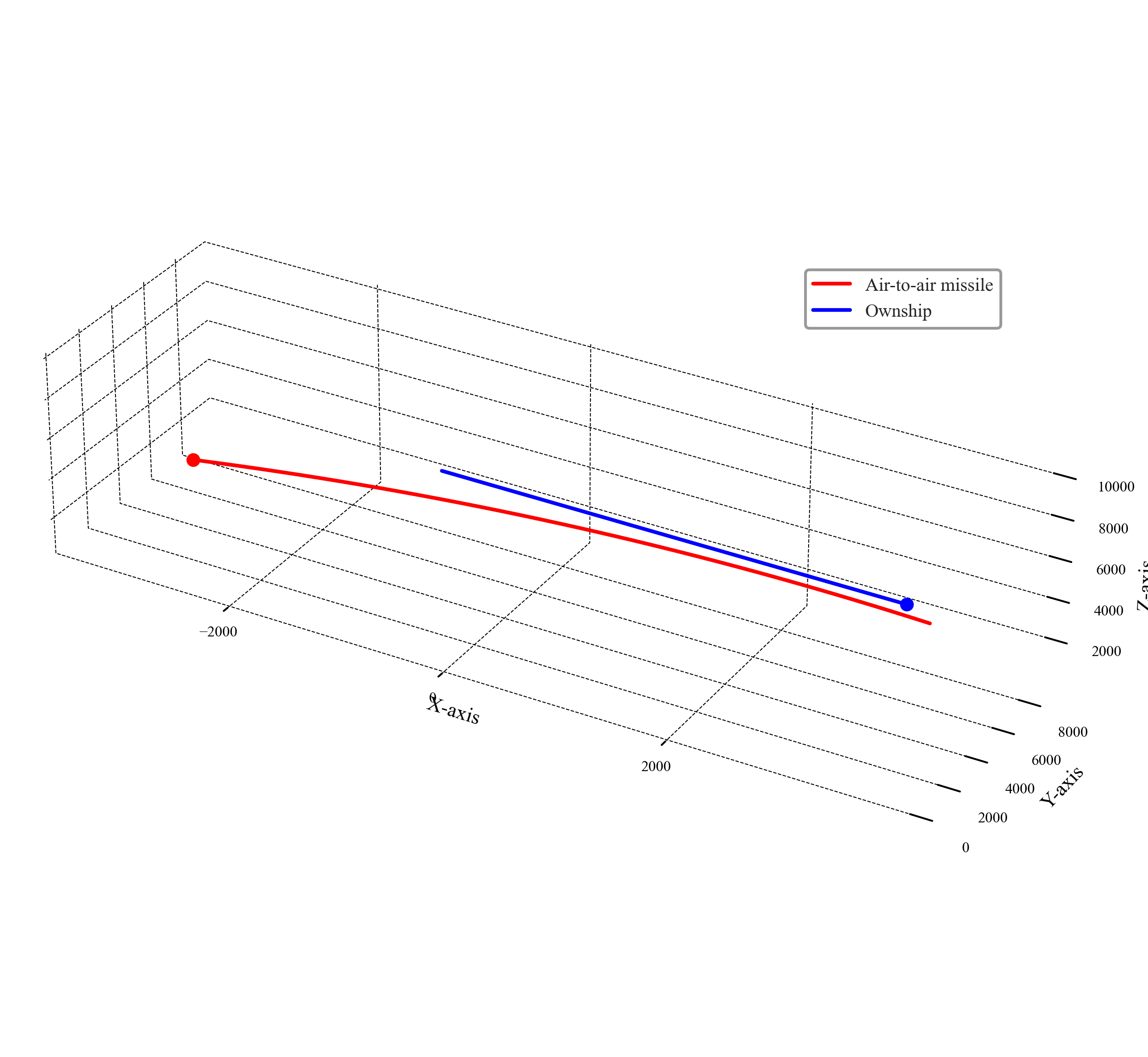}
        \caption{Trajectory: Air-to-air missile vs Ownship}
        \label{subfig:aam-vs-own}
    \end{subfigure}%
    \hfill
    \begin{subfigure}{0.32\textwidth}
        \centering
        \includegraphics[width=\linewidth]{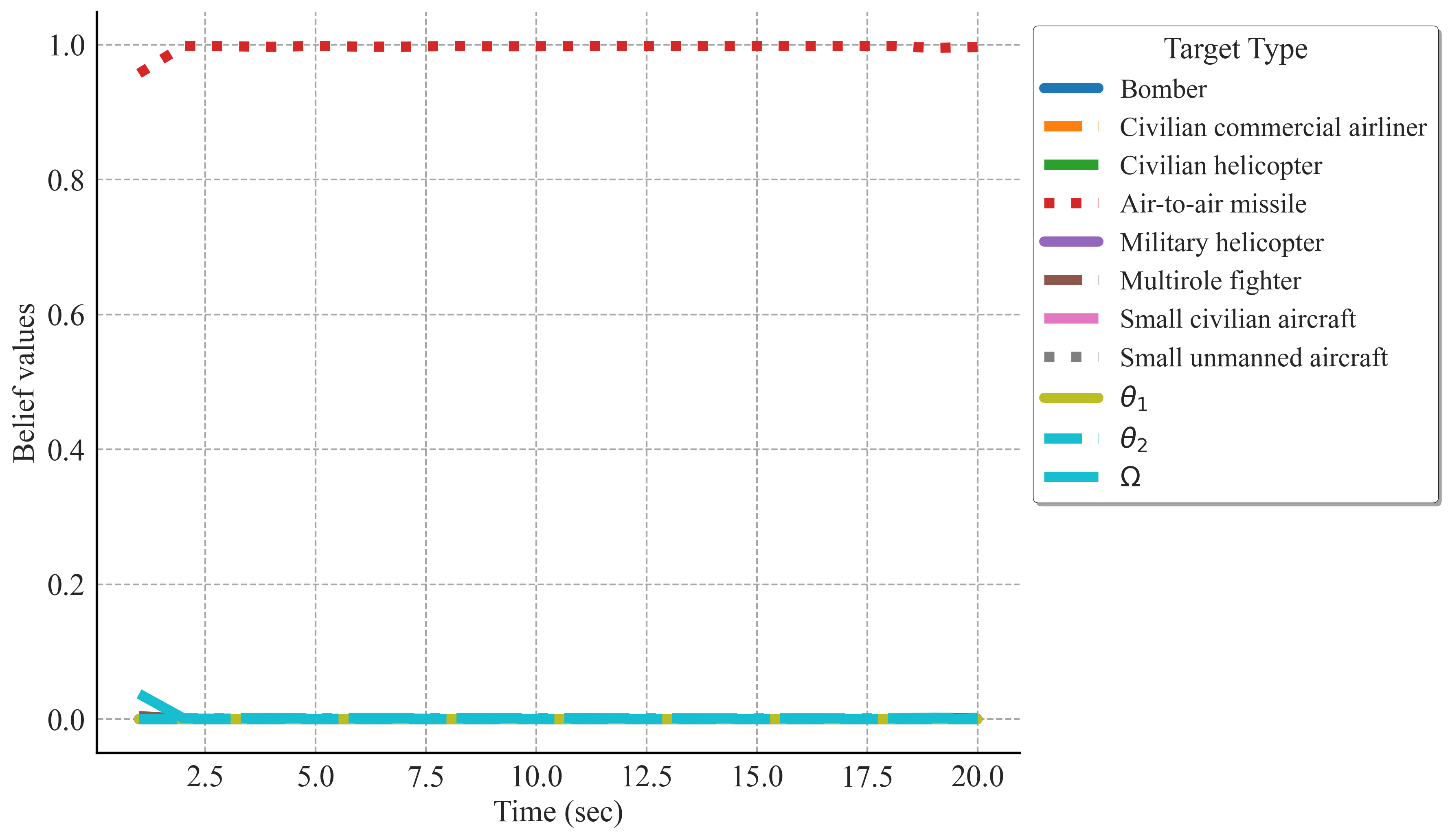}
        \caption{Type classification}
        \label{subfig:aam-vs-type}
    \end{subfigure}%
    \hfill
    \begin{subfigure}{0.32\textwidth}
        \centering
        \includegraphics[width=\linewidth]{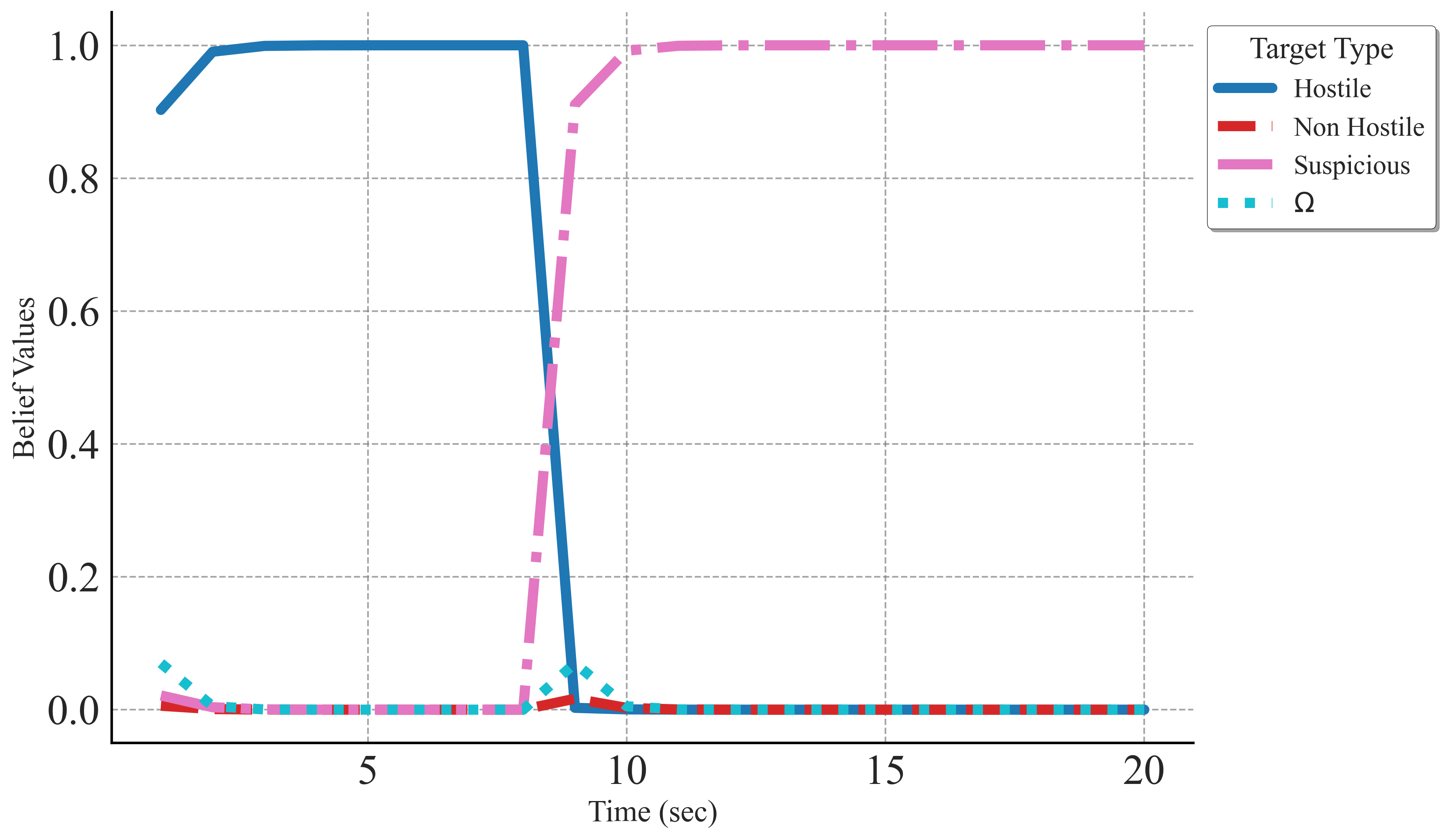}
        \caption{Intent prediction}
        \label{subfig:aam-vs-intent}
    \end{subfigure}

    \vspace{0.4em} 

    \begin{subfigure}{0.32\textwidth}
        \centering
        \includegraphics[width=\linewidth]{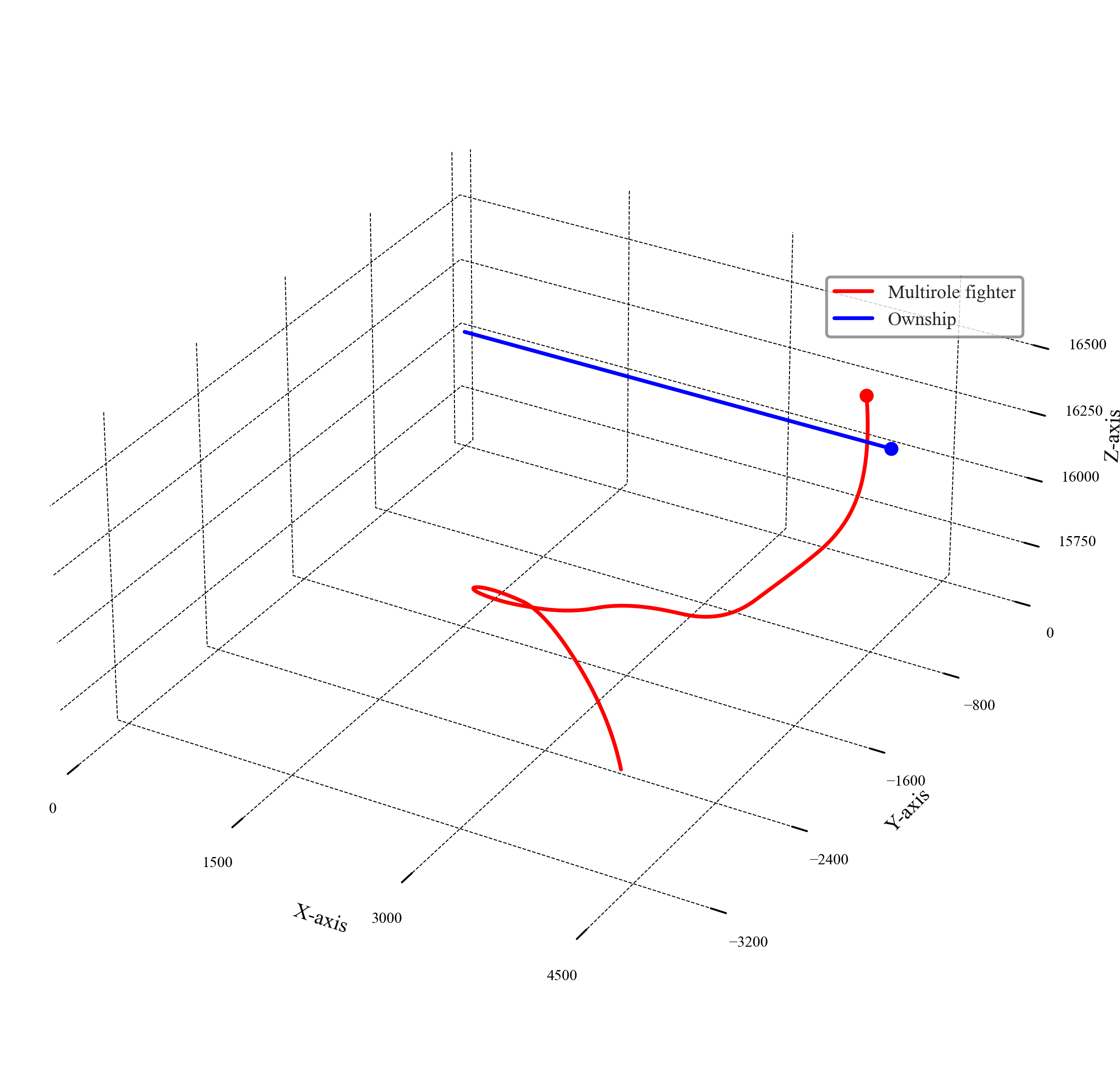}
        \caption{Trajectory: Multi-role fighter vs Ownship}
        \label{subfig:mf-vs-own}
    \end{subfigure}%
    \hfill
    \begin{subfigure}{0.32\textwidth}
        \centering
        \includegraphics[width=\linewidth]{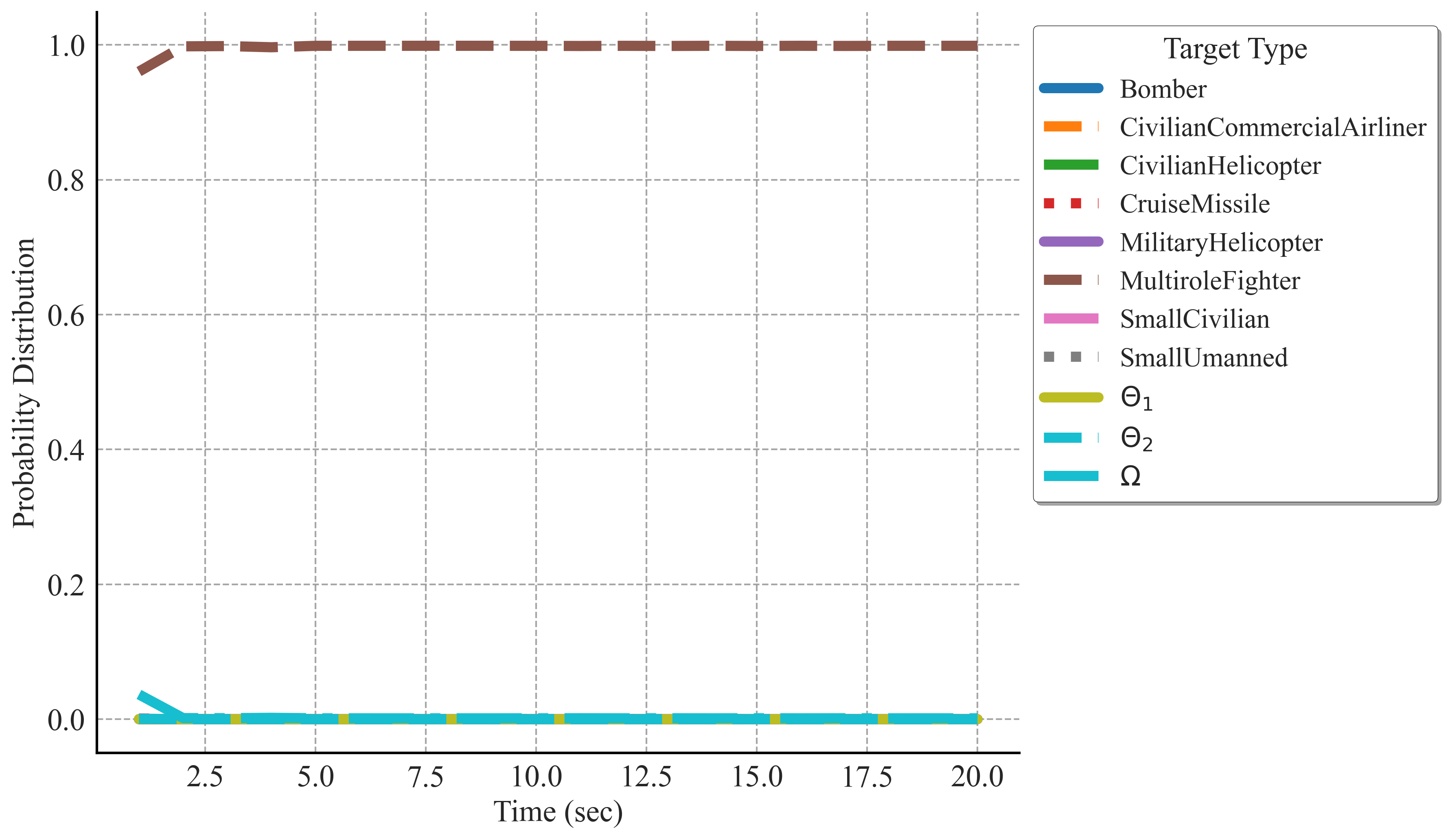}
        \caption{Type classification}
        \label{subfig:mf-vs-type}
    \end{subfigure}%
    \hfill
    \begin{subfigure}{0.32\textwidth}
        \centering
        \includegraphics[width=\linewidth]{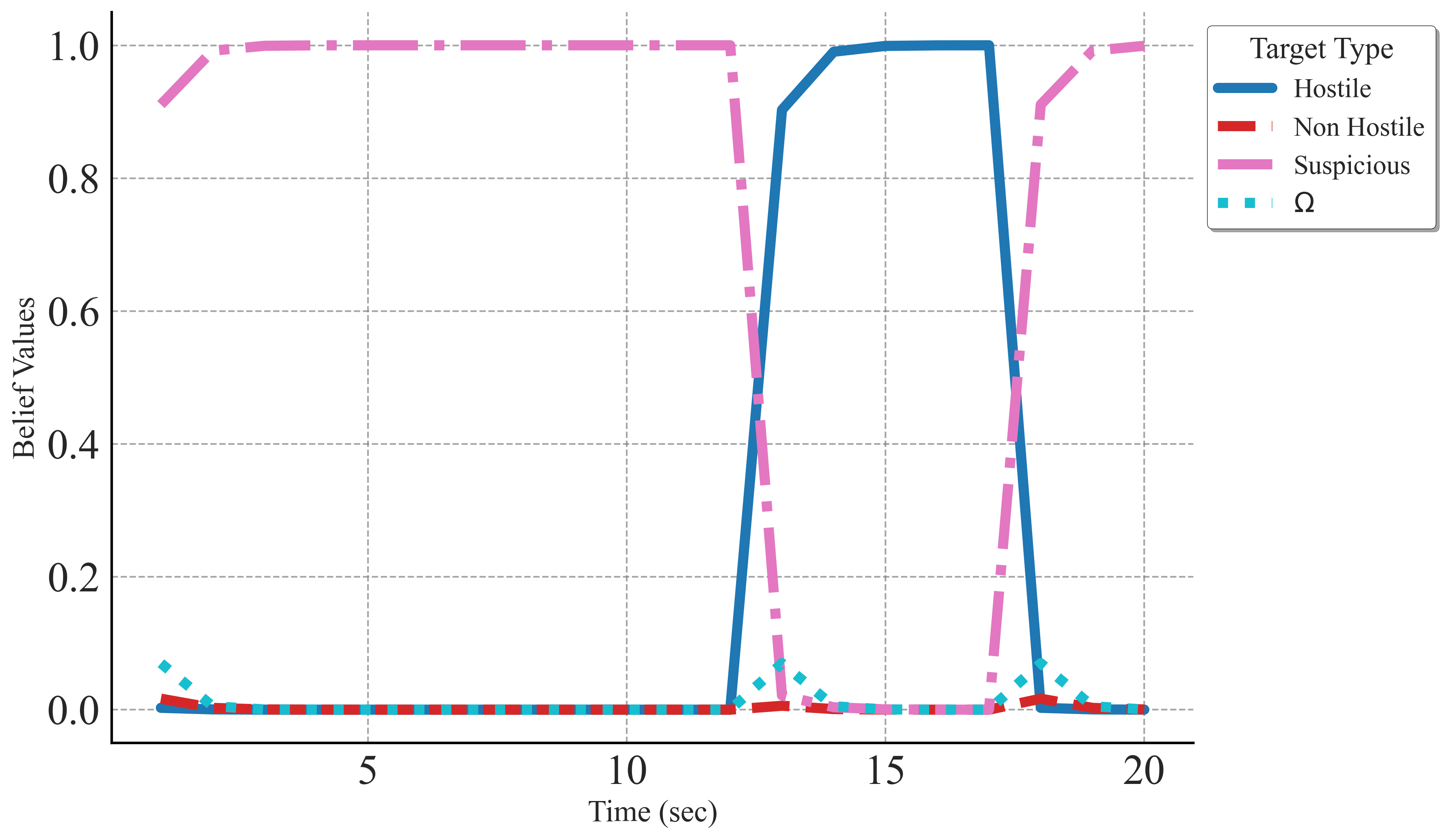}
        \caption{Intent prediction}
        \label{subfig:mf-vs-intent}
    \end{subfigure}

    \caption{Online classification results for (a–c) Air-to-air missile and (d–f) Multi-role fighter against ownship.}
    \label{fig:Online_Result}
\end{figure}
To demonstrate the effectiveness of the integrated target classification and intent prediction approach, two combat scenarios are considered, as shown in Fig.\ref{fig:Online_Result}. In the first scenario shown in Figs.\ref{subfig:aam-vs-own}-\ref{subfig:aam-vs-intent}, a one-to-one combat engagement between the ownship and an \texttt{Air-to-air missile} is presented. The corresponding classification results, also shown in Fig.\ref{subfig:aam-vs-type}, indicate that the \texttt{Air-to-air missile} is correctly classified, with the confidence in prediction increasing over time. The corresponding intent, depicted in Fig.\ref{subfig:aam-vs-intent}, is initially identified as \textit{Hostile} since the \texttt{Air-to-air missile} is approaching the ownship. However, once the target’s trajectory deviates away from the ownship, the intent is reassessed as \textit{Suspicious}. In the second scenario shown in Figs.\ref{subfig:mf-vs-own}-\ref{subfig:mf-vs-intent}, a combat engagement between the ownship and a \texttt{Multi-role fighter} is illustrated. As observed in Fig.\ref{subfig:mf-vs-type}, the target is correctly classified as a \texttt{Multi-role fighter}, with prediction confidence increasing over time. The intent is initially classified as \textit{Suspicious} because the does not display an \texttt{attacking} maneuver (Fig.~\ref{subfig:mf-vs-intent}). Later, as it begins to head towards the ownship, the intent is correctly reclassified as \textit{Hostile}. Subsequently, when the target moves away and no longer poses a threat, its predicted intent reverts to \textit{Suspicious} for the remainder of the observation period.

Finally, we analyze the execution time of Algorithm 1 to assess its computational efficiency. All reported run times correspond to the average per-subsample decision time (each subsample representing a 1-second segment), and the measured time reflects the total cost obtained after combining all processing stages—feature extraction, evaluation by the three classifiers, and fusion of their outputs.
The proposed approach is implemented in two widely used programming languages: Python and C++. As expected, the C++ implementation achieves faster execution, with an average runtime of $0.01,\mathrm{s}$ compared to $0.02,\mathrm{s}$ for Python. The experiments were conducted on a system equipped with an 11th Gen Intel(R) Core(TM) i9-12900 @ 2.50 GHz processor, an NVIDIA RTX 3060 GPU, and 32 GB RAM. The superior performance of C++ can be attributed to its compiled, low-level structure, which enables more efficient memory access and execution. This efficiency is especially critical in real-time applications, such as onboard aircraft platforms, where both processing speed and resource utilization are constrained. Beyond system-level evaluation, we also benchmarked the Python implementation across a range of computing platforms, from high-end to resource-limited devices, as shown in Fig.~\ref{fig:computa_python}. The results indicate that on high-end platforms, both classification time and memory usage are significantly reduced, highlighting the framework’s adaptability and efficiency across diverse computational environments. These findings not only validate the accuracy of the proposed method but also demonstrate its practical viability for real-time deployment in safety-critical systems.

\begin{figure}
    \centering
    \includegraphics[width=0.6\textwidth]{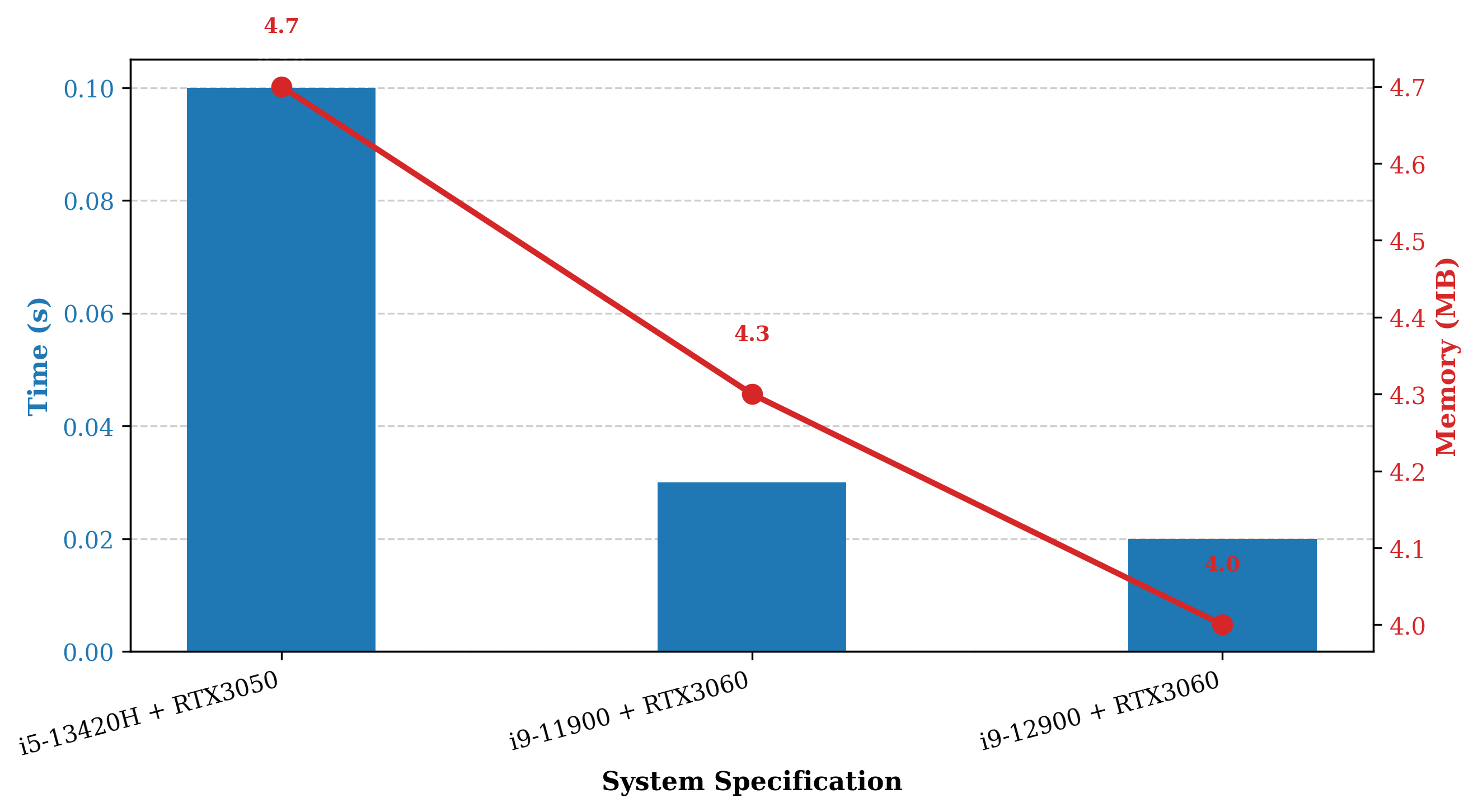}
    \caption{Computation time and memory usage across different computing platforms.}
    \label{fig:computa_python}
\end{figure}

\section{Conclusion and Future Work}
\label{sec:conclusion_and_future_works}
In this work, we address the challenge of aerial target classification and intent prediction under real-time constraints, where rapid and accurate decisions are essential. Conventional time-series classification approaches are unsuitable for combat scenarios, as they require full trajectory analysis and consequently introduce unacceptable delays. To overcome this limitation, we propose an integrated framework that enables target classification and intent prediction from partial data, supporting timely decision-making in situations that demand rapid tactical response. Instead of relying on complete time-series signals, predictions are generated on short sequential subsamples and refined through belief propagation. This approach enables both target type identification and behavioral intent prediction, categorizing aerial targets as hostile, non-hostile, or suspicious. To handle conflicting evidence in dynamic environments, we introduced a distance-based conflict-aware fusion mechanism that adaptively modulates the combination rule based on the level of disagreement among sources. Extensive simulations demonstrate that the proposed framework consistently outperforms traditional methods, maintaining high accuracy even in the presence of data ambiguity and enabling early detection of aerial targets. Beyond accuracy improvements, the core advantage of the framework lies in its ability to make stable, early predictions with limited observations—an essential requirement for real aerial combat, where pilots and autonomous systems must react within seconds. The conflict-aware fusion mechanism also offers a principled means to manage uncertainty, improving interpretability and trustworthiness of the predicted intent.Nevertheless, the current method assumes single-target engagement and does not explicitly model cooperative or deceptive behaviors, which are common in real combat situations. Sensor noise characteristics and maneuver unpredictability may also vary across platforms, requiring adaptive or learning-based calibration strategies. These limitations point toward promising research directions. Future work includes extending the framework to multi-target environments to support team-level tactical inference, incorporating online learning to adapt to evolving threat behaviors, and integrating domain-specific constraints from real combat operations to enhance deployment readiness.

\section*{Acknowledgment}
We sincerely thank the Aeronautical Development Agency, India, for funding this work and for providing valuable inputs on the behavior of various aerial targets. We greatly appreciate the contributions of Dr. Pavan Kumar Perepu, whose insights were invaluable to the successful completion of this study. We would also like to thank Kirtik Soni for implementing the proposed approach in C++. 

\bibliography{reference}

@article{SOUZA2018198,
title = {Asphalt pavement classification using smartphone accelerometer and Complexity Invariant Distance},
journal = {Engineering Applications of Artificial Intelligence},
volume = {74},
pages = {198-211},
year = {2018},
issn = {0952-1976},
doi = {https://doi.org/10.1016/j.engappai.2018.06.003},
url = {https://www.sciencedirect.com/science/article/pii/S0952197618301349},
author = {Vinicius M.A. Souza},
}

@inproceedings{frey2018f,
  title={F-35 information fusion},
  author={Frey, Thomas L and Aguilar, Chris and Engebretson, Kent and Faulk, David and Lenning, Layne G},
  booktitle={2018 Aviation Technology, Integration, and Operations Conference},
  pages={3520},
  year={2018}
}

@article{verbert2017bayesian,
  title={Bayesian and Dempster--Shafer reasoning for knowledge-based fault diagnosis--A comparative study},
  author={Verbert, Kim and Babu{\v{s}}ka, R and De Schutter, Bart},
  journal={Engineering Applications of Artificial Intelligence},
  volume={60},
  pages={136--150},
  year={201ftar7},
  publisher={Elsevier}
}

@inproceedings{ginoulhac2019target,
  title={Target classification based on kinematic data from AIS/ADS-B, using statistical features extraction and boosting},
  author={Ginoulhac, Rapha{\"e}l and Barbaresco, Fr{\'e}d{\'e}ric and Schneider, Jean-Yves and Pannier, Jean-Marie and Savary, S{\'e}bastien},
  booktitle={2019 20th International Radar Symposium (IRS)},
  pages={1--10},
  year={2019},
  organization={IEEE}
}

@inproceedings{ye2009time,
  title={Time series shapelets: a new primitive for data mining},
  author={Ye, Lexiang and Keogh, Eamonn},
  booktitle={Proceedings of the 15th ACM SIGKDD international conference on Knowledge discovery and data mining},
  pages={947--956},
  year={2009}
}

@article{hills2014classification,
  title={Classification of time series by shapelet transformation},
  author={Hills, Jon and Lines, Jason and Baranauskas, Edgaras and Mapp, James and Bagnall, Anthony},
  journal={Data mining and knowledge discovery},
  volume={28},
  pages={851--881},
  year={2014},
  publisher={Springer}
}

@article{schafer2015boss,
  title={The BOSS is concerned with time series classification in the presence of noise},
  author={Sch{\"a}fer, Patrick},
  journal={Data Mining and Knowledge Discovery},
  volume={29},
  pages={1505--1530},
  year={2015},
  publisher={Springer}
}

@article{gupta2020approaches,
  title={Approaches and applications of early classification of time series: A review},
  author={Gupta, Ashish and Gupta, Hari Prabhat and Biswas, Bhaskar and Dutta, Tanima},
  journal={IEEE Transactions on Artificial Intelligence},
  volume={1},
  number={1},
  pages={47--61},
  year={2020},
  publisher={IEEE}
}

@inproceedings{visser2003using,
  title={Using online learning to analyze the opponent’s behavior},
  author={Visser, Ubbo and Weland, Hans-Georg},
  booktitle={RoboCup 2002: Robot Soccer World Cup VI 6},
  pages={78--93},
  year={2003},
  organization={Springer}
}

@inproceedings{fukushima2018online,
  title={Online opponent formation identification based on position information},
  author={Fukushima, Takuya and Nakashima, Tomoharu and Akiyama, Hidehisa},
  booktitle={RoboCup 2017: Robot World Cup XXI 11},
  pages={241--251},
  year={2018},
  organization={Springer}
}

@article{smets2008transferable,
  title={The transferable belief model},
  author={Smets, Philippe and Kennes, Robert},
  journal={Classic Works of the Dempster-Shafer Theory of Belief Functions},
  pages={693--736},
  year={2008},
  publisher={Springer}
}

@book{shafer1976mathematical,
  title={A mathematical theory of evidence},
  author={Shafer, Glenn},
  volume={42},
  year={1976},
  publisher={Princeton university press}
}

@article{zadeh1986simple,
  title={A simple view of the Dempster-Shafer theory of evidence and its implication for the rule of combination},
  author={Zadeh, Lotfi A},
  journal={AI magazine},
  volume={7},
  number={2},
  pages={85--85},
  year={1986}
}

@article{kessentini2015dempster,
  title={A Dempster--Shafer theory based combination of handwriting recognition systems with multiple rejection strategies},
  author={Kessentini, Yousri and Burger, Thomas and Paquet, Thierry},
  journal={Pattern recognition},
  volume={48},
  number={2},
  pages={534--544},
  year={2015},
  publisher={Elsevier}
}

@article{cakir2008radar,
  title={Radar cross section (RCS) modeling and simulation, Part 2: A novel FDTD-based RCS prediction virtual tool for the resonance regime},
  author={Cakir, Gonca and Cakir, Mustafa and Sevgi, Levent},
  journal={IEEE Antennas and Propagation Magazine},
  volume={50},
  number={2},
  pages={81--94},
  year={2008},
  publisher={IEEE}
}

@article{cakir2014fdtd,
  title={An FDTD-based parallel virtual tool for RCS calculations of complex targets},
  author={Cakir, Gonca and Cakir, Mustafa and Sevgi, Levent},
  journal={IEEE Antennas and Propagation Magazine},
  volume={56},
  number={5},
  pages={74--90},
  year={2014},
  publisher={IEEE}
}

@inproceedings{sehgal2019automatic,
  title={Automatic target recognition using recurrent neural networks},
  author={Sehgal, Bharat and Shekhawat, Hanumant Singh and Jana, Sumit Kumar},
  booktitle={2019 International Conference on Range Technology (ICORT)},
  pages={1--5},
  year={2019},
  organization={IEEE}
}

@article{yager1987dempster,
  title={On the Dempster-Shafer framework and new combination rules},
  author={Yager, Ronald R},
  journal={Information sciences},
  volume={41},
  number={2},
  pages={93--137},
  year={1987},
  publisher={Elsevier}
}

@article{zhang2022information,
  title={An information fusion method based on deep learning and fuzzy discount-weighting for target intention recognition},
  author={Zhang, Zhuo and Wang, Hongfei and Geng, Jie and Jiang, Wen and Deng, Xinyang and Miao, Wang},
  journal={Engineering Applications of Artificial Intelligence},
  volume={109},
  pages={104610},
  year={2022},
  publisher={Elsevier}
}

@article{shin2018autonomous,
  title={An autonomous aerial combat framework for two-on-two engagements based on basic fighter maneuvers},
  author={Shin, Heemin and Lee, Jaehyun and Kim, Hyungi and Shim, David Hyunchul},
  journal={Aerospace Science and Technology},
  volume={72},
  pages={305--315},
  year={2018},
  publisher={Elsevier}
}

@article{dang2009coverage,
  title={Coverage-guided test generation for continuous and hybrid systems},
  author={Dang, Thao and Nahhal, Tarik},
  journal={Formal Methods in System Design},
  volume={34},
  pages={183--213},
  year={2009},
  publisher={Springer}
}

@article{zadeh1965fuzzy,
  title={Fuzzy sets},
  author={Zadeh, Lotfi A.},
  journal={Information and Control},
  volume={8},
  number={3},
  pages={338--353},
  year={1965},
  publisher={Elsevier}
}

@article{wang2023quick,
  title={Quick intention identification of an enemy aerial target through information classification processing},
  author={Wang, Yinhan and Wang, Jiang and Fan, Shipeng and Wang, Yuchen},
  journal={Aerospace Science and Technology},
  volume={132},
  pages={108005},
  year={2023},
  publisher={Elsevier}
}

@article{zhang2023combat,
  title={Combat Intention Recognition of Air Targets Based on 1DCNN-BiLSTM},
  author={Zhang, Chenhao and Zhou, Yan and Li, Hui and Liang, Futai and Song, Zihao and Yuan, Kai},
  journal={IEEE Access},
  volume={11},
  pages={134504--134516},
  year={2023},
  publisher={IEEE}
}

@article{gupta2019rule,
  title={Rule based classification of neurodegenerative diseases using data driven gait features},
  author={Gupta, Kartikay and Khajuria, Aayushi and Chatterjee, Niladri and Joshi, Pradeep and Joshi, Deepak},
  journal={Health and Technology},
  volume={9},
  pages={547--560},
  year={2019},
  publisher={Springer}
}

@inproceedings{dihel2023classifying,
  title={Classifying Aircraft using Velocity Data with Support Vector Machines and Likelihood Ratio Tests},
  author={Dihel, Logan and Dolph, Chester and Holbrook, Henry T and Roy, Sandip},
  booktitle={AIAA SCITECH 2023 Forum},
  pages={0898},
  year={2023}
}

@article{cemiloglu2025handling,
  title={Handling heterogeneity in Human Activity Recognition data by a compact Long Short Term Memory based deep learning approach},
  author={Cemiloglu, Ahmed and Akay, Bahriye},
  journal={Engineering Applications of Artificial Intelligence},
  volume={153},
  pages={110788},
  year={2025},
  publisher={Elsevier}
}

@article{jousselme2001new,
  title={A new distance between two bodies of evidence},
  author={Jousselme, Anne-Laure and Grenier, Dominic and Boss{\'e}, {\'E}loi},
  journal={Information fusion},
  volume={2},
  number={2},
  pages={91--101},
  year={2001},
  publisher={Elsevier}
}

@inproceedings{datagen,
   author = {Subash Kumaraguru and Devaprakash Muniraj},
   city = {Reston, Virginia},
   doi = {10.2514/6.2024-4582},
   isbn = {978-1-62410-716-0},
   booktitle = {AIAA AVIATION FORUM AND ASCEND 2024},
   month = {7},
   publisher = {American Institute of Aeronautics and Astronautics},
   title = {An Optimization-Based Approach for Generating Diverse Aerial Target Trajectories in a Combat Environment},
   year = {2024},
}

@article{wu2020benchmark,
  title={A benchmark data set for aircraft type recognition from remote sensing images},
  author={Wu, Zhi-Ze and Wan, Shou-Hong and Wang, Xiao-Feng and Tan, Ming and Zou, Le and Li, Xin-Lu and Chen, Yan},
  journal={Applied Soft Computing},
  volume={89},
  pages={106132},
  year={2020},
  publisher={Elsevier}
}

@article{singh2014dynamic,
  title={Dynamic classification of ballistic missiles using neural networks and hidden Markov models},
  author={Singh, Upendra Kumar and Padmanabhan, Vineet and Agarwal, Arun},
  journal={Applied Soft Computing},
  volume={19},
  pages={280--289},
  year={2014},
  publisher={Elsevier}
}

@article{wang2024intelligent,
  title={Intelligent recognition method of target tactical behavior intention in air combat based on deep learning},
  author={Wang, Xingyu and Yang, Zhen and Piao, Haiyin and Chai, Shiyuan and Huang, Jichuan and Zhou, Deyun},
  journal={Engineering Applications of Artificial Intelligence},
  volume={138},
  pages={109460},
  year={2024},
  publisher={Elsevier}
}

@article{aldhanhani2019framework,
  title={Framework for traffic event detection using Shapelet Transform},
  author={AlDhanhani, Ahmed and Damiani, Ernesto and Mizouni, Rabeb and Wang, Di},
  journal={Engineering Applications of Artificial Intelligence},
  volume={82},
  pages={226--235},
  year={2019},
  publisher={Elsevier}
}

@article{hui2017dempster,
  title={Dempster-Shafer evidence theory for multi-bearing faults diagnosis},
  author={Hui, Kar Hoou and Lim, Meng Hee and Leong, Mohd Salman and Al-Obaidi, Salah Mahdi},
  journal={Engineering Applications of Artificial Intelligence},
  volume={57},
  pages={160--170},
  year={2017},
  publisher={Elsevier}
}

@article{zhang2024target,
  title={A target intention recognition method based on information classification processing and information fusion},
  author={Zhang, Zhuo and Wang, Hongfei and Jiang, Wen and Geng, Jie},
  journal={Engineering Applications of Artificial Intelligence},
  volume={127},
  pages={107412},
  year={2024},
  publisher={Elsevier}
}

@article{oukhellou2010fault,
  title={Fault diagnosis in railway track circuits using Dempster--Shafer classifier fusion},
  author={Oukhellou, Latifa and Debiolles, Alexandra and Den{\oe}ux, Thierry and Aknin, Patrice},
  journal={Engineering Applications of Artificial Intelligence},
  volume={23},
  number={1},
  pages={117--128},
  year={2010},
  publisher={Elsevier}
}

\end{document}